\documentclass[11pt,letterpaper]{article}

\usepackage{graphicx,subfigure}

\def\showauthornotes{0}

\def\showkeys{0}
\def\showdraftbox{1}
\def\showcolorlinks{1}
\def\usemicrotype{1}
\def\showfixme{1}

\def\stocmode{0}
\def\jamesmode{0}
\def\arxivmode{0}
\def\fastmode{0}

\newcommand{\lca}{\mathsf{lca}}
\newcommand{\cost}{\mathsf{cost}}
\newcommand{\hst}{\vvmathbb{d}}
\newcommand{\llangle}{\left\langle}
\newcommand{\rrangle}{\right\rangle}
\newcommand{\vvW}{\vvmathbb{W}}

\ifnum\fastmode=0
\usepackage{etex}
\fi

\ifnum\fastmode=0
\usepackage[l2tabu, orthodox]{nag}
\fi

\usepackage{xspace,enumerate}

\usepackage[dvipsnames]{xcolor}

\ifnum\fastmode=0
\usepackage[T1]{fontenc}
\usepackage[full]{textcomp}
\fi

\usepackage[american]{babel}

\usepackage{mathtools}

\usepackage{amsthm}

\newtheorem{theorem}{Theorem}[section]
\newtheorem*{theorem*}{Theorem}

\newtheorem*{proposition*}{Proposition}
\newtheorem{lemma}[theorem]{Lemma}
\newtheorem*{lemma*}{Lemma}
\newtheorem{corollary}[theorem]{Corollary}
\newtheorem*{conjecture*}{Conjecture}

\newtheorem*{fact*}{Fact}

\newtheorem*{exercise*}{Exercise}

\newtheorem*{hypothesis*}{Hypothesis}

\theoremstyle{definition}
\newtheorem{definition}[theorem]{Definition}

\newtheorem{exercise-easy}[theorem]{Exercise}
\newtheorem{exercise-med}[theorem]{Exercise}
\newtheorem{exercise-hard}[theorem]{Exercise$^\star$}

\newtheorem*{claim*}{Claim}

\newtheorem*{remark*}{Remark}

\newtheorem*{observation*}{Observation}

\usepackage[
letterpaper,
top=0.7in,
bottom=0.9in,
left=1in,
right=1in]{geometry}

\ifnum\arxivmode=0
\usepackage{newpxtext} %
\usepackage{textcomp} %
\usepackage[varg,bigdelims]{newpxmath}
\usepackage[scr=rsfso]{mathalfa}%
\usepackage{bm} %
\let\mathbb\varmathbb
\fi

\ifnum\arxivmode=1
\usepackage[varg]{pxfonts} %
\usepackage{textcomp} %
\usepackage[scr=rsfso]{mathalfa}%
\usepackage{bm} %
\fi

\ifnum\showkeys=1
\usepackage[color]{showkeys}
\fi

\makeatletter
\@ifpackageloaded{hyperref}
{}
{\usepackage[pagebackref]{hyperref}}

\definecolor{bleudefrance}{rgb}{0.01, 0.1, 1.0}
\definecolor{azure}{rgb}{0.0, 0.5, 1.0}

\ifnum\showcolorlinks=1
\hypersetup{
pagebackref=true,
colorlinks=true,
urlcolor=blue,
linkcolor=blue,
citecolor=OliveGreen}
\fi

\ifnum\showcolorlinks=0
\hypersetup{
pagebackref=true,
colorlinks=false,
pdfborder={0 0 0}}
\fi

\usepackage{prettyref}

\newcommand{\savehyperref}[2]{\texorpdfstring{\hyperref[#1]{#2}}{#2}}

\newrefformat{eq}{\savehyperref{#1}{\textup{(\ref*{#1})}}}
\newrefformat{lem}{\savehyperref{#1}{Lemma~\ref*{#1}}}
\newrefformat{def}{\savehyperref{#1}{Definition~\ref*{#1}}}
\newrefformat{thm}{\savehyperref{#1}{Theorem~\ref*{#1}}}
\newrefformat{cor}{\savehyperref{#1}{Corollary~\ref*{#1}}}
\newrefformat{cha}{\savehyperref{#1}{Chapter~\ref*{#1}}}
\newrefformat{sec}{\savehyperref{#1}{Section~\ref*{#1}}}
\newrefformat{app}{\savehyperref{#1}{Appendix~\ref*{#1}}}
\newrefformat{tab}{\savehyperref{#1}{Table~\ref*{#1}}}
\newrefformat{fig}{\savehyperref{#1}{Figure~\ref*{#1}}}
\newrefformat{hyp}{\savehyperref{#1}{Hypothesis~\ref*{#1}}}
\newrefformat{alg}{\savehyperref{#1}{Algorithm~\ref*{#1}}}
\newrefformat{rem}{\savehyperref{#1}{Remark~\ref*{#1}}}
\newrefformat{item}{\savehyperref{#1}{Item~\ref*{#1}}}
\newrefformat{step}{\savehyperref{#1}{step~\ref*{#1}}}
\newrefformat{conj}{\savehyperref{#1}{Conjecture~\ref*{#1}}}
\newrefformat{fact}{\savehyperref{#1}{Fact~\ref*{#1}}}
\newrefformat{prop}{\savehyperref{#1}{Proposition~\ref*{#1}}}
\newrefformat{prob}{\savehyperref{#1}{Problem~\ref*{#1}}}
\newrefformat{claim}{\savehyperref{#1}{Claim~\ref*{#1}}}
\newrefformat{relax}{\savehyperref{#1}{Relaxation~\ref*{#1}}}
\newrefformat{red}{\savehyperref{#1}{Reduction~\ref*{#1}}}
\newrefformat{part}{\savehyperref{#1}{Part~\ref*{#1}}}
\newrefformat{ex}{\savehyperref{#1}{Exercise~\ref*{#1}}}
\newrefformat{property}{\savehyperref{#1}{Property~\ref*{#1}}}

\newcommand{\Sref}[1]{\hyperref[#1]{\S\ref*{#1}}}

\usepackage{nicefrac}

\ifnum\fastmode=0
\ifnum\usemicrotype=1
\usepackage{microtype}
\fi
\fi

\ifnum\showauthornotes=2
\newcommand{\mynotes}[1]{{\sffamily\small\color{teal}{#1}}\medskip}
\newcommand{\Authornote}[2]{{\sffamily\small\color{blue}{[#1: #2]}}\medskip}
\newcommand{\Authornotecolored}[3]{{\sffamily\small\color{#1}{[#2: #3]}}}
\newcommand{\Authorcomment}[2]{{\sffamily\small\color{gray}{[#1: #2]}}}
\newcommand{\Authorstartcomment}[1]{\sffamily\small\color{gray}[#1: }

\newcommand{\Authorfnote}[2]{\footnote{\color{red}{#1: #2}}}
\newcommand{\Authorfixme}[1]{\Authornote{#1}{\textbf{??}}}
\newcommand{\Authormarginmark}[1]{\marginpar{\textcolor{red}{\fbox{\Large #1:!}}}}
\newcommand{\myexplain}[1]{{\sffamily\small\color{red}{\noindent [Explanation:\medskip\newline \begin{quote}#1\hfill]\end{quote}}}\medskip}
\newcommand{\explain}[1]{{\sffamily\small\color{red}{#1}}\medskip}

\else

\newcommand{\mynotes}[1]{}
\newcommand{\Authornote}[2]{}
\newcommand{\Authornotecolored}[3]{}
\newcommand{\Authorcomment}[2]{}
\newcommand{\Authorstartcomment}[1]{}

\newcommand{\Authorfnote}[2]{}
\newcommand{\Authorfixme}[1]{}
\newcommand{\Authormarginmark}[1]{}
\newcommand{\myexplain}[1]{}
\newcommand{\explain}[1]{}

\ifnum\showauthornotes=1
\renewcommand{\myexplain}[1]{{\sffamily\small\color{red}{\noindent \begin{quote}{\bf Explanation:} \medskip\newline #1\end{quote}}}\medskip}
\fi

\fi

\newcommand{\jnote}{\Authornote{J}}

\ifnum\showfixme=0

\fi

\usepackage{boxedminipage}

\newcommand{\norm}[1]{\lVert#1\rVert}

\newcommand{\Esymb}{\mathbb{E}}
\newcommand{\Psymb}{\mathbb{P}}

\DeclareMathOperator*{\E}{\Esymb}

\DeclareMathOperator*{\ProbOp}{\Psymb}

\renewcommand{\Pr}{\ProbOp}

\newcommand{\textparen}[1]{\text{(#1)}}

\ifx\because\undefined
\newcommand{\because}[1]{\textparen{because #1}}
\else
\renewcommand{\because}[1]{\textparen{because #1}}
\fi

\newcommand{\seteq}{\mathrel{\mathop:}=}

\newcommand\bdot\bullet

\ifx\mathds\undefined %
\newcommand{\Ind}{\mathbb I}
\else
\newcommand{\Ind}{\mathds 1}
\fi

\DeclareMathOperator{\opt}{opt}

\DeclareMathOperator{\dist}{dist}

\newcommand{\Z}{\mathbb Z}
\newcommand{\N}{\mathbb N}
\newcommand{\R}{\mathbb R}

\newcommand{\cA}{\mathcal A}

\newcommand{\cC}{\mathcal C}

\newcommand{\cE}{\mathcal E}

\newcommand{\cL}{\mathcal L}

\newcommand{\cP}{\mathcal P}

\newcommand{\cR}{\mathcal R}

\newcommand{\cT}{\mathcal T}

\renewcommand{\leq}{\leqslant}

\renewcommand{\geq}{\geqslant}

\ifnum\showdraftbox=1

\else

\fi

\let\epsilon=\varepsilon

\numberwithin{equation}{section}

\newcommand\MYcurrentlabel{xxx}

\newcommand{\MYstore}[2]{%
  \global\expandafter \def \csname MYMEMORY #1 \endcsname{#2}%
}

\newcommand{\MYload}[1]{%
  \csname MYMEMORY #1 \endcsname%
}

\newcommand{\MYnewlabel}[1]{%
  \renewcommand\MYcurrentlabel{#1}%
  \MYoldlabel{#1}%
}

\newcommand{\MYdummylabel}[1]{}

\newcommand{\torestate}[1]{%
  \let\MYoldlabel\label%
  \let\label\MYnewlabel%
  #1%
  \MYstore{\MYcurrentlabel}{#1}%
  \let\label\MYoldlabel%
}

\newcommand{\restatetheorem}[1]{%
  \let\MYoldlabel\label
  \let\label\MYdummylabel
  \begin{theorem*}[Restatement of \prettyref{#1}]
    \MYload{#1}
  \end{theorem*}
  \let\label\MYoldlabel
}

\newcommand{\restatelemma}[1]{%
  \let\MYoldlabel\label
  \let\label\MYdummylabel
  \begin{lemma*}[Restatement of \prettyref{#1}]
    \MYload{#1}
  \end{lemma*}
  \let\label\MYoldlabel
}

\newcommand{\restateprop}[1]{%
  \let\MYoldlabel\label
  \let\label\MYdummylabel
  \begin{proposition*}[Restatement of \prettyref{#1}]
    \MYload{#1}
  \end{proposition*}
  \let\label\MYoldlabel
}

\newcommand{\restatefact}[1]{%
  \let\MYoldlabel\label
  \let\label\MYdummylabel
  \begin{fact*}[Restatement of \prettyref{#1}]
    \MYload{#1}
  \end{fact*}
  \let\label\MYoldlabel
}

\newcommand{\restate}[1]{%
  \let\MYoldlabel\label
  \let\label\MYdummylabel
  \MYload{#1}
  \let\label\MYoldlabel
}

\newcommand{\addreferencesection}{
  \phantomsection
\ifnum\stocmode=0
  \addcontentsline{toc}{section}{References}
\else
  \addcontentsline{toc}{section}{References \hspace*{1in} --------- End of extended abstract ---------}
\fi

}

\newcommand{\e}{\epsilon}

\let\origparagraph\paragraph
\renewcommand{\paragraph}[1]{\vspace*{-10pt}\origparagraph{#1.}}

\allowdisplaybreaks

\usepackage{paralist}

\usepackage{comment}

\let\pref=\prettyref

\newcommand*{\id}{\mathrm{id}}

\newcommand{\diam}{\mathrm{diam}}

\newcommand{\vertiii}[1]{{\left\vert\kern-0.25ex\left\vert\kern-0.25ex\left\vert #1 
          \right\vert\kern-0.25ex\right\vert\kern-0.25ex\right\vert}}

\renewcommand{\Ind}{\vvmathbb{1}}
\ifnum\jamesmode=0
\fi
\let\1\Ind

\DeclareMathOperator{\argmin}{\mathrm{argmin}}

\usepackage{accents}

\usepackage{enumitem}
\usepackage[normalem]{ulem}

\newcommand{\Lip}{\mathrm{lip}}

\newcommand{\len}{\mathrm{len}}

\renewcommand{\Z}{\vvmathbb{Z}}
\renewcommand{\R}{\vvmathbb{R}}

\begin{document}

\title{$k$-server via multiscale entropic regularization}
\author{Sebastien Bubeck \\ {\small Microsoft Research} \and Michael B. Cohen \\ {\small MIT} \and James R. Lee \\ {\small University of Washington} \and
Yin Tat Lee \\ {\small University of Washington} \and Aleksander M\k{a}dry \\ {\small MIT}}

\date{}

\maketitle

\begin{abstract}
   We present an $O((\log k)^2)$-competitive randomized algorithm for the $k$-server problem on hierarchically separated trees (HSTs).
   This is the first $o(k)$-competitive randomized algorithm
   for which the competitive ratio is independent of the size of the underlying HST.
   Our algorithm is designed in the framework of online mirror descent where the mirror map is a multiscale entropy.
   When combined with Bartal's static HST embedding reduction, this leads to an $O((\log k)^2 \log n)$-competitive algorithm on any $n$-point metric space.
   We give a new dynamic HST embedding that yields an $O((\log k)^3 \log \Delta)$-competitive algorithm on any metric space where the ratio of the
   largest to smallest non-zero distance is at most $\Delta$.
\end{abstract}

\begingroup
\hypersetup{linktocpage=false}
\setcounter{tocdepth}{2}
\tableofcontents
\endgroup

\newpage

\section{Introduction}

Perhaps the most widely-studied problem in the field
of online algorithms and competitive analysis is the {\em $k$-server problem,}
introduced in \cite{MMS90} to generalize
and abstract a number of related problems arising
in the study of paging and caching.
The problem has been the object of intensive study
since its inception, motivated largely by two long-standing
conjectures about the competitive ratios that
can be achieved by deterministic
and randomized algorithms, respectively.

We recall the problem briefly;
see \pref{sec:prelims} for a formal
definition of the model.
Fix a metric space $(X,d)$ and $k \geq 1$,
as well as
an initial placement $\rho_0 \in X^k$ of $k$ servers in $X$.
An online $k$-server algorithm operates as follows.
At each time step, a request $r_t \in X$
comes online, and the algorithm must respond to this
request by moving one of the servers to $r_t$
(unless there is already a server there).
The cost of the algorithm is the total distance
moved by all the servers over the course of the request sequence.
An {\em offline} algorithm
operates in the same manner, but is allowed access to the entire request sequence
in advance.  An online algorithm has competitive ratio
$\alpha$ if, for every request sequence, its movement cost per unit time step
is within an $\alpha$ factor of that achieved
by the optimal offline algorithm.

\medskip
\noindent
{\bf Randomization.}
The authors of \cite{MMS90} stated the {\em $k$-server conjecture:}  On
any metric space with at least $k+1$ points, the best competitive ratio
achieved by {\em deterministic} online algorithms is precisely $k$.
They showed that the ratio is always at least $k$.
While the conjecture is still open in general, Koutsoupias and Papadimitriou
resolved it within a factor of two:  The work function algorithm
obtains a competitive ratio of $2k-1$ on any metric space \cite{KP95}.
We refer the reader to the book \cite{BE98} for further background
on online algorithms and the $k$-server problem.

In the context of the {\em $k$-paging problem}, which is the special
case of $k$-server on a metric space with all distances in the set $\{0,1\}$,
it was observed that randomness can help an online algorithm dramatically:
It is known that the competitive ratio for $k$-paging is precisely the
$k$th harmonic number $H_k$ for every $k \geq 1$ \cite{FKL91,MS91}.

\medskip
\noindent
{\bf Hierarchically separated trees.}
There is another class of metric spaces on which one can prove lower bounds
on the competitive ratio even for randomized algorithms.
Consider a finite, rooted tree $\cT=(V,E)$ with vertex weights $w : V \to \R_+$
that are non-increasing along every root-leaf path.  Let $\cL \subseteq V$
denote the set of leaves and define the metric
\[
   d_{w}(\ell,\ell') \seteq w_{\lca(\ell,\ell')} \qquad \forall \ell,\ell'\in \cL\,.
\]
In this case, $d_w$ is an ultrametric (in fact, all finite ultrametrics are of this form).
If $w_v \leq w_u/\tau$ whenever $v$ is a child of $u$, then one says that
$(\cT,w)$ is a {\em $\tau$-hierarchically separated tree ($\tau$-HST)}
and the metric space $(\cL,d_{w})$ is referred to as a {\em $\tau$-HST metric.}
By showing lower bounds on every HST and additionally demonstrating
that all sufficiently large metric spaces contain subsets that are
close to HSTs, the authors of \cite{BBM06} proved a lower bound
on the competitive ratio for every metric space.  Using the metric Ramsey bounds
from \cite{BLMN05}, it holds that the competitive ratio for
randomized algorithms is at least $\Omega\left(\frac{\log k}{\log \log k}\right)$ on every
metric space with more than $k$ points.

A lack of further lower bounds motivates the folklore ``Randomized $k$-server conjecture'' (see, e.g.,
\cite{BBK99}) that the randomized competitive ratio is $(\log k)^{O(1)}$ on every metric space,
or even $O(\log k)$, matching the lower bound for $k$-paging.
One can also consult
the survey \cite{Kou09} (in particular, Conjecture 2 and the surrounding discussion).

Since the seminal work of Bartal on probabilistic embeddings of finite
metric spaces into HSTs \cite{Bar96,Bar98}, it has been understood that
obtaining {\em upper bounds} on the competitive ratio for HSTs is of central importance.
Indeed, the competitive ratio for $n$-point metric spaces is at most an $O(\log n)$
factor larger than the competitive ratio for $n$-point HST metrics.
This bound follows from Bartal's approach
combined with the optimal distortion estimate of \cite{FRT04}.

The authors of \cite{CMP08} give an $O(\log D)$-competitive randomized
algorithm on binary $2$-HSTs with combinatorial depth at most $D$, and
in \cite{BBMN15}, a major breakthrough was achieved when the authors exhibited
an $(\log n)^{O(1)}$-competitive algorithm for general $n$-vertex HSTs.
In the present work, we obtain a competitive ratio independent of the
size of the underlying HST, thereby verifying a long-held belief.

\begin{theorem} \label{thm:mainhst}
   For every $k \geq 2$, there is an $O((\log k)^2)$-competitive randomized
   algorithm for the $k$-server problem on any HST.
\end{theorem}

As mentioned previously, this yields an $O((\log k)^2 \log n)$-competitive
randomized algorithm for every $n$-point metric space,
via 
probabilistic embeddings of finite
metric spaces into distributions over HSTs.  The embedding
underlying this reduction is {\em oblivious} to the request sequence.
While this is a very convenient feature for the analysis, oblivious
embeddings cannot avoid losing an $\Omega(\log n)$ factor
in the competitive ratio.
We show that in certain cases, this can be circumvented
via the use of {\em dynamic} HST embeddings
where the embedding is allowed to depend on the request sequence.

\begin{theorem}
   For every $k \geq 2$ and every finite metric space $(X,d)$,
there is an $O\left((\log k)^3 \log (1+\cA_X)\right)$-competitive randomized
algorithm for the $k$-server problem on $(X,d)$, where
   \[
      \cA_X \seteq \frac{\max_{x,y \in X} d(x,y)}{\min_{x \neq y \in X} d(x,y)}\,.
   \]
\end{theorem}

\subsection{Mirror descent and entropic regularization}

Our algorithm is most naturally stated in the framework of 
continuous-time {\em mirror descent.} 
This framework was originally introduced for convex optimization in \cite{NY83} (see also \cite{Bub15}),
and recently it has played a key role in online decision making; see, e.g., \cite{Haz16} for the online learning setting,
and \cite{ABBS10, BCN14} for applications to metrical task systems. Typically an entropy functional is used as a mirror map,
and a key contribution of our work is to propose an appropriate multiscale entropy functional.

We establish some properties of a general setup in \pref{sec:traverse}
and, as a warmup application,
present in \pref{sec:paging}
an $O(\log k)$-competitive algorithm for the (fractional)
weighted paging problem that is closely related to the
algorithm of \cite{BBN12}.
This already exhibits a couple key ideas in a simplified setting,
including the natural use of the Bregman divergence as a potential function,
and the utility of using $k+\e$ servers for some $\e < 1$.
In \pref{sec:trees}, we begin transferring these ideas to the
setting of the $k$-server problem on trees.
Notably, the polytope underlying our state representation
is the one derived from the fractional allocation problem
as employed in \cite{CMP08} and \cite{BBMN15}.
In \pref{sec:depth}, we introduce the crucial
idea of an auxilliary potential function that
tracks the weighted depth of the underlying fractional server measure,
and in \pref{sec:weighted-depth} we show how using
a time-varying weight can be leveraged to obtain
an $O((\log k)^2)$ competitive ratio.

\subsection{Preliminaries}
\label{sec:prelims}

We use the notations $\R_+ \seteq [0,\infty)$ and $\Z_+ = \Z \cap \R_+$.
If $X$ and $Y$ are two metric spaces and $F : X \to Y$ is Lipschitz, we use $\|F\|_{\Lip}$
to denote the Lipschitz constant of $F$.
Consider a bounded, complete metric space $(X,d)$ and two Borel probability measures $\mu$ and $\nu$ on $X$.
We use $\vvW_X^1(\mu,\nu)$ to denote the $L^1$-transportation
distance between $\mu$ and $\nu$ (sometimes called the Earthmover metric):
\[
	\vvW_X^1(\mu,\nu) \seteq \inf \E\left[d(Y,Y')\right],
\]
where the infimum is over all jointly distributed random variables $(Y,Y')$
such that $Y$ has marginal $\mu$ and $Y'$ has marginal $\nu$.
The definition is extended in the natural way to any two Borel measures
satisfying $\mu(X)=\nu(X)$.

\medskip
\noindent
{\bf Online algorithms and the competitive ratio.}
Let $(X,d)$ be a metric space and fix $k \geq 1$.
We now describe the $k$-server problem more formally.
The input is a sequence $\langle \sigma_t \in X : t \geq 1\rangle$ of {\em requests.}
At every time $t$, an online algorithm maintains
a state $\rho_t \in X^k$
which can be thought of as the location of $k$ {\em servers} in the space $X$.
At time $t$, the algorithm is required to have a server at the requested site $r_t \in X$.
In other words, a feasible state $\rho_t$ is one that {\em services $r_t$:}
\[
   r_t \in \left\{\left(\rho_t\right)_1, \ldots, \left(\rho_t\right)_k\right\}.
\]
Formally, an {\em online algorithm} is a sequence of mappings
$\bm{\rho} = \langle \rho_1, \rho_2, \ldots, \rangle$
where, for every $t \geq 1$,
$\rho_t : X^t \to X^k$ maps a request sequence $\langle r_1, \ldots, r_t\rangle$
to a $k$-server state that services $r_t$.
In general, $\rho_0 \in X^k$ will denote some initial state of the algorithm.

The {\em cost of the algorithm $\bm{\rho}$ in servicing $\bm{r} = \langle r_t : t \geq 1\rangle$} is defined as
the sum of the movements of all the servers:
\[
\cost_{\bm{\rho}}(\bm{r};k) \seteq \sum_{t \geq 1} d_{X^k}\left(\rho_t(r_1, \ldots, r_t), \rho_{t-1}(r_1,\ldots,r_{t-1})\right)\,,
\]
where 
\[
   d_{X^k}\left((x_1, \ldots, x_k), (y_1, \ldots, y_k)\right) \seteq \sum_{i=1}^k d_X(x_i,y_i) \qquad \forall x_1, \ldots,x_k,y_1,\ldots,y_k \in X\,,
\]
and $\rho_0 \in X^k$ is some fixed initial configuration.

For a given request sequence $\bm{r} = \langle r_t : t \geq 1\rangle$,
denote the cost of the {\em offline optimum} by
\[
   \cost^*(\bm{r};k) \seteq \inf_{\langle \rho_1,\rho_2,\ldots\rangle} \sum_{t \geq 1} d_{X^k}\left(\rho_t,\rho_{t-1}\right)\,,
\]
where the infimum is over all sequences $\langle \rho_1,\rho_2,\ldots\rangle$ such that $\rho_t$ services $r_t$ for each $t \geq 1$.

A {\em randomized online algorithm $\bm{\rho}$} is a random online algorithm
that is feasible with probability one.  Such an algorithm is said to be {\em $\alpha$-competitive}
if for every $\rho_0 \in X^k$, there is a constant $c > 0$ such that for all $\bm{r}$:
\[
\E\left[\cost_{\bm{\rho}}(\bm{r};k)\right] \leq \alpha \cdot \cost^*(\bm{r};k) + c\,.
\]

\section{Traversing a convex body online}
\label{sec:traverse}

Suppose that $\mathsf{K} \subseteq \R^n$ is a closed convex set and $f : \R_+ \times \mathsf{K} \to \R^n$
is a time-varying vector field defined on $\mathsf{K}$.  It is very natural
to consider the {\em projected dynamical system:}
\begin{align*}
   x'(t) &= \Pi_\mathsf{K}\left(x(t), f(t,x(t))\right) \\
   x(0) &= x_0\,,
\end{align*}
where
\[
   \Pi_\mathsf{K}(x,v) \seteq \lim_{\e \to 0} \frac{P_\mathsf{K}(x+\e v)-x}{\e}\,,
\]
and
\[
   P_\mathsf{K}(x) \seteq \argmin \left\{\|x-z\|_2 : z \in \mathsf{K}\right\}\,.
\]

One can interpret this as trying to ``flow'' along the vector field in direction $f(t,x(t))$ 
while being confined to remain in the convex body $\mathsf{K}$.
But since projection is a discontinuous operation (imagine hitting the boundary of a polytope, for instance),
the classical theory of existence and uniqueness of ODEs no longer applies.
Fortunately, there is now a well-established theory for projected dynamical systems.

Let us denote the {\em normal cone to $\mathsf{K}$ at $x$} by
\[
   N_\mathsf{K}(x) \seteq \left\{ p \in \R^n : \langle p,y-x\rangle \leq 0 \textrm{ for all } y \in \mathsf{K}\right\}\,.
\]
If $\mathsf{K}$ is a polyhedron, then the normal cone to $\mathsf{K}$ at $x$ is the cone spanned by the normals
of the tight constraints at $x$.

\begin{lemma}
   \label{lem:normal_polytope}Given any matrix $A\in\R^{m\times n}$
   and $b\in\R^{m}$, consider the polyhedron $\mathsf{K} \seteq \{x\in \R^n:Ax\leq b\}$.
   For any $x\in \mathsf{K}$, it holds that
   \[
      N_{\mathsf{K}}(x)=\left\{A^{T}y\ :\ y\geq0\text{ and }y^{T}(b-Ax)=0\right\}.
   \]
\end{lemma}

For our applications, we will want to consider a projected dynamical system
with respect to a non-Euclidean geometry on $\mathsf{K}$.  Let $\Phi : \mathsf{K} \to \R^n$
be a strongly convex function.  Then $\nabla^2 \Phi$ is positive definite
and can be thought of as a Riemannian metric on $\mathsf{K}$.
One can describe the possible dynamics in the geometry induced by $\Phi$ by the {\em differential inclusion}
\begin{align}
   \nabla^2 \Phi(x(t)) x'(t) &\in f\!\left(t,x(t)\right) - N_\mathsf{K}(x(t))\label{eq:inclusion}\\
   x(0) &= x_0\nonumber \\
   x(t) &\in \mathsf{K} \qquad \forall t \geq 0\,.\nonumber
\end{align}
Note that the right-hand side of \eqref{eq:inclusion} is a {\em set} of vectors,
and a solution $x(t)$ is one that satisfies the inclusion.
It turns out that, under suitable assumptions, an absolutely continuous solution exists,
and under stronger assumptions, the solution is unique.
The following theorem is proved in \pref{sec:mirror}.

\begin{theorem}\label{thm:mirror-intro}
   Consider a compact convex set $\mathsf{K} \subseteq \R^n$, a strongly convex function $\Phi:\mathsf{K}\rightarrow\R$,
   and a continuous function $f:[0,\infty)\times \mathsf{K}\rightarrow\R^n$. Suppose furthermore
   that $\nabla^{2}\Phi(x)^{-1}$ is continuous.  Then for any $x_{0}\in \mathsf{K}$,
   there is an absolutely continuous solution $x : [0,\infty) \to \mathsf{K}$
   satisfying
   \begin{align*}
      x'(t) & \in \nabla^{2}\Phi(x(t))^{-1}\left(f(t,x(t))-N_{\mathsf{K}}(x(t))\right), \\
      x(0) & =x_{0}.\nonumber 
   \end{align*}
   If we further assume that $\nabla^{2}\Phi(x)$ is Lipschitz and $f$
   is locally Lipschitz, then the solution is unique.
\end{theorem}

\subsection{Evolution of the Bregman divergence}

Recall that the {\em Bregman divergence} associated to $\Phi : \mathsf{K} \to \R$ is given by
\[
   D_{\Phi}(y;x) \seteq \Phi(y) - \Phi(x) - \langle \nabla \Phi(x), y-x\rangle\,.
\]
We will use $D_{\Phi}$ as a potential function to track the ``discrepancy'' between our algorithm
and the optimal offline algorithm.
In fact, it will be slightly easier to work with the function
\[
   \hat{D}_{\Phi}(y;x) \seteq - \Phi(x) - \langle \nabla \Phi(x),y-x\rangle\,.
\]

Suppose now that $x(t)$ is an absolutely continuous solution to
the differential inclusion \eqref{eq:inclusion} and write
\[
   \nabla^2 \Phi(x(t)) \partial_t x(t) = f\!\left(t,x(t)\right) - \lambda(t)
\]
with $\lambda(t) \in N_{\mathsf{K}}(x(t))$.
One concludes immediately that for $y \in \mathsf{K}$:
      \begin{align}
         \partial_t \hat{D}_{\Phi}(y; x(t)) &= \langle \nabla^2 \Phi(x(t)) \partial_t x(t), x(t)-y\rangle 
         \nonumber \\
         &=  \llangle f\!\left(t,x(t)\right) - \lambda(t), x(t) - y\rrangle \nonumber \\
         &\leq \llangle f\!\left(t,x(t)\right), x(t)-y\rrangle\,, \label{eq:lower-breg}
      \end{align}
where in the last inequality we have used
that $\langle \lambda(t), y-x(t)\rangle \leq 0$ since $\lambda(t) \in N_{\mathsf{K}}(x(t))$.

\subsection{Application: Fractional weighted paging}
\label{sec:paging}

\newcommand{\relint}{\mathrm{relint}}

\jnote{Describe weighted paging problem.}

Fix $k \geq 1$.
Consider the fractional weighted paging problem on pages in $[n]$
with a cache of size $k$ and positive weights $\{ w_i > 0 : i \in [n]\}$.
For $z \in \R^n$, define the weighted $\ell_1$ norm:
\[
   \|z\|_{\ell_1(w)} \seteq \sum_{i=1}^n w_i \left|z_i\right|,
\]
and the dual norm:
\[
   \|z\|_{\ell_{\infty}(1/w)} \seteq \max \left\{ \frac{|z_i|}{w_i} : i \in [n]\right\}.
\]
Note that if $z(t) \in \R^n$ is an online algorithm for $t \in [0,T]$, then the
movement cost is precisely
\[
   \int_0^T \|\partial_t z(t)\|_{\ell_1(w)}\,dt\,.
\]
Moreover, up to a factor of two, we can charge our algorithm
only for the cost of moving fractional mass {\em into} a node, i.e.,
\begin{equation}\label{eq:cost-into}
   \int_0^T \left\|\left(\partial_t z(t)\right)_+\right\|_{\ell_1(w)}\,dt\,,
\end{equation}
where for $z\in \R^n$, we denote $\left(z\right)_+ \seteq \left(\max(0,z_1),\ldots,\max(0,z_n)\right)$.

\jnote{Describe charging ALG only for the incoming movement.}

\subsubsection{Entropy-regularized dynamics}

\jnote{We will charge the algorithm for moving a page into cache, not out.  Is equivalent.}

Define the {\em fractional $k$-antipaging polytope}
\[
   \mathsf{P} \seteq \left\{ x \in [0,1]^n : \sum_{i=1}^n x_i = n-k\right\}.
\]
Here, we think of $1-x_i$ as the fractional amount of page $i$ that sits in the cache (hence
$x_i$ is the amount of fractional ``antipage'').
Define also the entropic regularizer
\[
   \Phi(x) \seteq \sum_{i=1}^n w_i x_i \log x_i\,.
\]

\jnote{Define relative interior}

Suppose the current state of the cache is described by
a point $x \in \mathsf{P}$ such that $x_i(0) > 0$ for all $i \in [n]$.
When a request $r \in \{1,\ldots,n\}$ is received,
we need to decrease $x_r$ to $0$.
To this end, we use the (constant) control function:
\[
   f\!\left(t,x(t)\right) \seteq -e_r \qquad \forall t \geq 0\,.
\]

Now the intended
trajectory is given by the differential inclusion:
\[
   \nabla^2 \Phi(x(t)) \partial_t x(t) \in - e_r- N_{\mathsf{P}}(x(t))\,.
\]

Let us analyze the dynamics which are described by
\begin{equation}\label{eq:paging-dynamics}
   \nabla^2 \Phi(x(t)) \partial_t x(t) = -e_r - \lambda(t)\,,
\end{equation}
where $\lambda(t) \in N_{\mathsf{P}}(x(t))$.
From \pref{lem:normal_polytope},
it is an exercise to compute that
\begin{equation}\label{eq:paging-lambda}
   \lambda(t) = \sum_{i=1}^n \lambda_i(t) e_i - \mu(t) \1
\end{equation}
for some $\{\lambda_i(t) \geq 0\}$ such that $\lambda_i(t) > 0 \implies x_i(t) < 1$.

Here, the $\{\lambda_i\}$ functions are the Lagrangian multipliers for
the constraints $\{x_i \leq 1\}$ of $\mathsf{P}$, and $\mu$ is the multiplier
for $\sum_{i=1}^n x_i = n-k$.
The fact that $x_i(t) > 0$ for $t > 0$
is implicitly enforced by $\Phi$ (assuming
some boundedness on the control $f$).
Since $\|\nabla \Phi(x)\|_{2} \to \infty$ as $x$ approaches the boundary
of the positive orthant, $\Phi$ acts as a barrier preventing
the evolution from leaving $\R_+$.
We do not stress this point formally at the moment since
we will soon need to maintain a more restrictive condition.

Let $\hat{x} \in \{0,1\}^n \cap \mathsf{P}$ denote an integral antipaging point
with $\hat{x}_r = 0$ (i.e., a state which has satisfied the request).
Then \eqref{eq:lower-breg} immediately yields
\begin{equation}\label{eq:lower-paging}
   \partial_t \hat{D}_{\Phi}\left(\hat{x}; x(t)\right) \leq \langle e_r, \hat{x} - x(t)\rangle = - x_r(t)\,.
\end{equation}

Moreover,
from \eqref{eq:paging-dynamics} and \eqref{eq:paging-lambda},
one easily calculates:
\begin{equation}\label{eq:traj}
   \partial_t x_i(t) = \frac{x_i(t)}{w_i} \left(-e_r(i) - \lambda_i(t) + \mu(t)\right)\,,
\end{equation}
where
\[
   \lambda_i(t) = \begin{cases}
      0 & x_i(t) < 1 \\
      \mu(t) & \textrm{otherwise.}
   \end{cases}\,.
\]
Using $\partial_t \sum_{i=1}^n x_i(t) = 0$ yields
\begin{equation}\label{eq:muvalue}
   \mu(t) = \frac{x_r(t)/w_r}{\sum_{i : x_i(t) < 1} x_i(t)/w_i}\,.
\end{equation}
In particular, if $x_r(t) < 1$ then $\mu(t) \leq 1$, hence
from \eqref{eq:traj},
the instantaneous movement cost (recall \eqref{eq:cost-into}) is bounded by
\[
   w_r |\partial_t x_r(t)| \leq x_r(t)\,.
\]
Thus the potential change \eqref{eq:lower-paging} compensates for the movement cost.

Now we have to address convergence, and here we run into a problem:  $\hat{D}_{\Phi}(\hat{x};x(t))$ could
be infinite!
Therefore \eqref{eq:lower-paging} does not show that $x_r(t) \to 0$ as $t \to \infty$.

\subsubsection{Moving in the interior}

Our solution to this problem will be to shift the variables away from the boundary of $\R_+$.
For $\delta > 0$, define
\[
   \mathsf{P}_{\delta} \seteq \mathsf{P} \cap [\delta,1]^n = \left\{ x\in [\delta,1]^n : \sum_{i=1}^n x_i = n -k \right\}\,.
\]
Clearly we cannot remain in this polytope and still service a request $r$ by moving
to a point with $x_r=0$.  Instead, we will allow our algorithm to satisfy the weaker
constraint $x_r=\delta$, and then afterward show that any such algorithm
can be transformed---in an online manner---to a valid fractional paging algorithm,
as long as $\delta$ is chosen small enough.
Furthermore, we can easily ensure that our dynamics remain inside $P_{\delta}$
by simplying stopping when $x_r(T)=\delta$ (if we can ensure that
there is a time $T$ at which this occurs).

Now note that
\begin{equation}\label{eq:grad-paging}
   \sup_{x \in \mathsf{P}_{\delta}} \|\nabla \Phi(x)\|_{\ell_\infty(1/w)} \leq O\!\left(\log \frac{1}{\delta}\right)\,.
\end{equation}
Thus if we ensure that $x(t) \in \mathsf{P}_{\delta}$, then \eqref{eq:lower-paging}
implies that $x_r(T)=\delta$ occurs after some finite time $T$.

We are left to analyze how the potential changes when OPT moves.
But from the definition and \eqref{eq:grad-paging}, we have
\[
   \partial_t \hat{D}_{\Phi}(y(t); x) = \langle \nabla \Phi(x), - \partial_t y(t)\rangle \leq O(\log \tfrac{1}{\delta}) \|\partial_t y(t)\|_{\ell_1(w)}\,.
\]
Thus we obtain an $O(\log \frac{1}{\delta})$-competitive algorithm, where $\delta$ is the smallest
constant such that we can round (online) a fractional $\frac{k}{1-\delta}$-paging algorithm to a genuine
fractional $k$-paging algorithm.  As we will see now, 
this can be done when $\delta = \frac{1}{2k}$.

\paragraph{Transforming to a valid fractional paging algorithm}

Consider a request sequence $\vec{r} = (r_1, r_2, \ldots, r_M)$, and
a differentiable map $x : [0,T] \to \mathsf{P}_{\delta}$
that services $\vec{r}$ in the sense that
there are times $t_1 < t_2 < \cdots < t_M$ such that
$x_{r_i}(t_i) = \delta$.
We may assume that $x(t)$ is {\em elementary}
in the sense that for almost every $t \in [0,T]$:
\[
   \partial_t x(t) = (e_i-e_j)\,dt
\]
for some $i \neq j \in [n]$.

Define
\[
   z(t) \seteq \frac{1-x(t)}{1-\delta} \qquad \forall t \in [0,T]\,.
\]
Then $z_{r_i}(t_i)=1$, so $z$ represents a trajectory on
measures that services $\vec{r}$, but problematically
we have $\|z(t)\|_1 = \frac{k}{1-\delta}$ for $t \in [0,T]$.

We fix this as follows:  Let $\e \seteq \frac{\delta k}{1-\delta}$ and
define $\sigma : \R_+ \to \R_+$ so that $\sigma|_{[\ell,\ell+\e]} = \ell$ for every $\ell \in \Z_+$
and $\sigma$ is extended affinely to the rest of $\R_+$.
For $z \in \R_+^n$, define $\sigma(z) \seteq (\sigma(z_1),\ldots,\sigma(z_n))$, and 
consider the trajectory $\sigma(z(t))$ for $t \in [0,T]$.
Observe first that
\[
   \int_0^T \left\|\partial_t \sigma(z(t))\right\|_{\ell_1(w)}\,dt \leq 
   \|\sigma\|_{\Lip} \int_0^T \left\|\partial_t z(t)\right\|_{\ell_1(w)}\,dt
   \leq \frac{1}{1-\e}\int_0^T \left\|\partial_t z(t)\right\|_{\ell_1(w)}\,dt\,.
\]
Thus for $\delta = \frac{1}{2k}$, the movement cost has increased
under $\sigma$ by only an $O(1)$ factor.

Because $\sigma$ is superadditive, it also holds that for every $t \in [0,T]$,
\[
   \sum_{i=1}^n \sigma(z_i(t)) \leq \sigma\left(\sum_{i=1}^n z_i(t)\right) = \sigma\left(k+\e\right) = k\,.
\]
Therefore we use at most $k$ fractional server mass at any point in time.  We are left to show
that, at no additional movement cost, this can be transformed into an algorithm
that maintains fractional server mass exactly $k$.

To that end, we may assume that $\|\sigma(z(0))\|_1 = k$.
It will be easiest to think of a weighted star metric on vertices $V = \{1,2,\ldots,n\} \cup \{0\}$,
where $0$ is the center of the star and the edge $(0,i)$ has length $w_i$.
When $\partial_t \sigma(z(t)) = -e_i\, dt$ for some $i \in [n]$,
the instantaneous movement cost of $\sigma(z(t))$ is $w_i e_i\, dt$.
Instead of deleting this mass, we can move it to $0$ for the same cost.
Similarly, when $\partial_t \sigma(z(t)) = e_i\,dt$,
the instantaneous movement cost is again $w_i e_i\,dt$ and
instead of creating mass, we can move $dt$ mass from $0$ to $i$.

\section{$k$-server on trees}
\label{sec:trees}
\newcommand{\vvT}{\vvmathbb{T}}
\newcommand{\vvM}{\vvmathbb{M}}
\newcommand{\ch}{\chi}
\newcommand{\extmeas}{\widehat{M}}
\renewcommand{\root}{\vvmathbb{r}}

Consider a rooted tree $\cT = (V,E)$ with root $\root \in V$
and leaves $\cL \subseteq V$.
Let $\{w_v \geq 0 : v\in V\}$ be a collection of nonnegative
weights on $V$ with $w_{\root} = 0$.
We will suppose that every leaf $\ell \in \cL$ is at the same
combinatorial distance from the root.

For $v \in V$, let $\cL(v) \subseteq \cL$ denote the set of leaves
beneath $v$.
For $u \in V \setminus \{\root\}$, let $p(u) \in V$ denote the parent of $u$, and
write $\vec{E} = \{ \vec{e} : e \in E \} = \left \{ (p(u),u) : u \in V \setminus \{\root\}\right\}$
for the set of edges directed away from the root.
For $(u,v) \in \vec{E}$,
define $\len_w(u,v) \seteq w_v$.
Let $\dist_w(x,y)$ denote the weighted path distance between $x,y \in V$,
where an edge $e \in E$ is given weight $\len_w(\vec{e})$.
We say that the pair $(\cT,w)$ is a {\em $\tau$-adic HST}
if for every $v \in V \setminus \{\root\}$,
it holds that $w_v = \tau^{j}$ for some $j \in \Z$
and, moreover, if $(u,v) \in \vec{E}$ then $w_v = w_u/\tau$.

For $z \in \R^V$ and $w \in \R_+^V$, denote
\[
   \left\|z\right\|_{\ell_1(w)} \seteq \sum_{v \in V} w_v \left|z_v\right|\,.
\]

\paragraph{Leaf measures, internal measures, and supermeasures}

A {\em leaf measure} is a point $z \in \R_+^{\cL}$.
The {\em mass} of a leaf measure is defined as the quantity $\sum_{\ell \in \cL} z_{\ell}$.
An {\em internal supermeasure} is a point $z \in \R_+^V$ such that
\begin{equation}\label{eq:supmeas}
   z_u \geq \sum_{v : (u,v) \in \vec{E}} z_v \qquad \forall u \in V\,.
\end{equation}
The {\em mass} of an internal supermeasure is the quantity $z_{\root}$.
We say that $z\in \R_+^V$ is an {\em internal measure} if \eqref{eq:supmeas}
is satisfied with equality.

For a leaf measure $z \in \R_+^{\cL}$, we define its {\em lifting to an internal measure}
by
\[
   \hat{z}_v \seteq \sum_{\ell \in \cL(v)} z_{\ell} \qquad \forall v \in V\,.
\]
Let $\extmeas$ denote the set of internal measures on $V$.
It is straightforward to see that this is precisely the class of lifted leaf measures.

\begin{lemma}\label{lem:W1norm}
   For leaf measures $y,z \in \R_+^{\cL}$ with $\|y\|_1 = \|z\|_1$, it holds that 
   \[
      \vvW^1_{w}(y,z) = \left\|\hat{y}-\hat{z}\right\|_{\ell_1(w)}\,.
   \]
\end{lemma}

A {\em fractional $k$-server algorithm for $(\cL,\dist_w)$}
is an online sequence $\llangle z^{(t)} \in \R_+^{\cL} : t=0,1,2,\ldots\rrangle$
of leaf measures of mass $k$ such that for every $t \geq 1$:
$z_{\ell_t}^{(t)} \geq 1$ if $\ell_t$ is the requested leaf at time $t$.
We also require that $z^{(0)}$ is integral.
The {\em cost} of such an algorithm is defined by
\[
   \sum_{t \geq 0} \vvW^1_w\left(z^{(t)}, z^{(t+1)}\right)\,.
\]

\begin{lemma}[{\cite[\S 5.2]{BBMN15}}]
   The following holds for all $\tau > 5$.
   If $(X,d_X)$ is a $\tau$-HST metric
   and there is a fractional online $k$-server algorithm
   for $(X,d_X)$, then there is
   a randomized integral online $k$-server algorithm whose expected
   cost is at most $O(1)$ times larger.
\end{lemma}

\subsection{$k+\e$ fractional servers}
\label{sec:keservers}

For the remainder of the proof, we will work with continuous time trajectories $z : [0,T] \to \R_+^{\cL}$
whose movement cost is measured by
\[
   \int_0^T \left\|\partial_t \hat{z}_t\right\|_{\ell_1(w)}\,dt\,,
\]
in light of \pref{lem:W1norm}.  Obviously such a trajectory
can be mapped to a discrete-time algorithm by choosing times $T_1 \geq T_2 \geq \cdots$
that correspond to discrete times $t=1,2,\ldots$.

\begin{lemma}\label{lem:internal}
   If $y : [0,T] \to \R_+^V$ is a trajectory taking values in internal supermeasures of mass $k$, then there
   is an (adapted) trajectory $z : [0,T] \to \R_+^{\cL}$ taking values in leaf measures of mass $k$
   such that:
   \begin{enumerate}
      \item For every leaf $\ell \in \cL$:  $z_t(\ell) \geq y_t(\ell)$, and
      \item $\vvW^1_{w}(z(0),z(T)) \leq \int_0^T \|\partial_t y(t)\|_{\ell_1(w)}\,dt$.
   \end{enumerate}
\end{lemma}

\begin{proof}
This lemma follows from a more general principle:  If $(X,d)$ is a metric space and $X' \subseteq X$,
then an online (fractional) $k$-server algorithm on $(X,d)$ servicing a sequence of requests in $X'$
can be converted to an online (fractional) $k$-server algorithm on $(X',d|_{X'\times X'})$
without increasing the movement cost.  This is a straightforward consequence of the triangle inequality.

Now observe that we can envision every trajectory on internal supermeasures $y(t) \in \R_+^V$ with $t \in [0,T]$
as a trajectory on genuine measures $\tilde{y}(t) \in \R_+^V$ defined by
\[
   \tilde{y}_{u}(t) \seteq y_u(t) - \sum_{v : (u,v) \in \vec{E}} y_u(t) \qquad \forall u \in V\,.
\]
And moreover,
\[
   \vvW^1_w(\tilde{y}(0),\tilde{y}(T)) \leq \int_0^T \|\partial_t y(t)\|_{\ell_1(w)}\,.
\]
Since $(\cL,\dist_w)$ is a subspace of $(V,\dist_w)$, this completes the proof
by our earlier observation.
\end{proof}

\begin{lemma}\label{lem:kpluseps}
For every $0 \leq \e <1$, a fractional online $(k+\e)$-server algorithm on $(\cL,\dist_{w})$
can be converted to a fractional online $k$-server algorithm so that the movement cost
increases by a factor of at most $\frac{1}{1-\e}$.
\end{lemma}

\begin{proof}
The proof is similar to the case for fractional paging.
Define $\sigma:\R_{+}\to\R_{+}$ so that $\sigma|_{[\ell,\ell+\e]}=\ell$
for every $\ell\in\Z_{+}$ and $\sigma$ is extended affinely to the
rest of $\R_{+}$.
For $y\in\R_{+}^{V}$, define $\sigma(y)\seteq(\sigma(y_{v}))_{v\in V}$.

Consider a trajectory $z : [0,T] \to \R_+^{\cL}$ taking values in leaf measures of mass $k+\e$.
Then since $\sigma$ is superadditive, it holds that $\sigma(\hat{z}(t))$
is an internal supermeasure for every $t \in [0,T]$.
Moreover, $\sigma(\hat{z}_{\root}(t)) = \sigma(k+\e)=k$, so $\sigma(\hat{z}(t))$ is
an internal supermeasure of mass $k$.

Finally, note that
\[
\int_{0}^{T}\left\Vert \partial_{t}\sigma(\hat{z}(t))\right\Vert _{\ell_{1}(w)}\,dt \leq\|\sigma\|_{\Lip}\int_{0}^{T}\left\Vert \partial_{t}\hat{z}(t)\right\Vert _{\ell_{1}(w)}\,dt
\leq\frac{1}{1-\e}\int_{0}^{T}\left\Vert \partial_{t}\hat{z}(t)\right\Vert_{\ell_{1}(w)}\,dt\,.
\]
Now applying \pref{lem:internal} to the internal supermeasures $\sigma(\hat{z}(t))$ completes the proof.
\end{proof}

In light of the preceding lemma, it suffices to construct a competitive fractional $(k+\e)$-server algorithm with $\e < 1$
for any request sequence on $(\cL,\dist_w)$.

\subsection{The assignment polytope}

For $u \in V$, write
\[
   \ch(u) \seteq \left \{ (v,j) : (u,v) \in \vec{E}, j \in [N_u] \right\}\,
\]
where $N_u$ is the number of leaves in the subtree rooted at $u$.
Denote \[\Lambda \seteq \left\{\vphantom{\bigoplus} (\root,i) : i \in [N_{\root}]\right\} \cup \bigcup_{u \in V} \chi(u)\,.\]
With a slight abuse of notation, we sometimes write $\sum_{i\geq1} f(x_{u,i})$ instead of $\sum_{i \in [N_u]} f(x_{u,i})$.
The {\em assignment polytope on $\cT$} is defined by
\begin{align*}
   \mathsf{A} \seteq \left\{ \vphantom{\int} x \in [0,1]^{\Lambda} :\right. &\sum_{i \leq |S|} x_{u,i} \leq \sum_{(v,j) \in S} x_{v,j} \quad \forall u \in V, S \subseteq \ch(u)\,, \\
                                                                                  & \left. x_{\root,i} = \1_{\{i > k\}}\quad\forall i \geq 1 \vphantom{\int}\right\}\,.
\end{align*}

Fix some $\delta > 0$ and define
the shifted multiscale entropy by
\[
   \Phi(x) \seteq \sum_{u \in V} w_u \sum_{i \geq 1} (x_{u,i}+\delta) \log (x_{u,i}+\delta)\,.
\]
For $\ell \in \cL$, denote $x_{\ell} \seteq x_{\ell,1}$.
Finally, for $u \in V$ and $x\in \mathsf{A}$, define the {\em associated server measure} $z \in \R_+^V$ by
\[
   z_u \seteq \frac{1}{1-\delta} \sum_{i \geq 1} (1-x_{u,i})\,.
\]

Suppose we receive a request at $\ell \in \cL$.
Let $x : [0,\infty) \to \mathsf{A}$ be an absolutely continuous trajectory satisfying
\begin{equation}\label{eq:traj2}
   \partial_t x(t) = \nabla^2 \Phi(x(t))^{-1} (- e_{\ell,1} - \lambda(t))
\end{equation}
with $\lambda(t) \in N_{\mathsf{A}}(x(t))$ for almost every $t \geq 0$.
Denote
\[
   T = \inf \left\{ t \geq 0 : x_{\ell}(t) < \delta \right\}.
\]

By construction, as long as $T < \infty$ (see \pref{lem:kserver-breg} below),
it holds that
$\{z_u(t) : t \in[0,T]\}$ is a fractional $\frac{k}{1-\delta}$-server trajectory that
services the request at $\ell$.
Now set $\delta \seteq \frac{1}{2k}$ so that $\frac{k}{1-\delta} < k+1$ for $k \geq 2$.

\subsection{Dynamics}

We now describe in detail the dynamics of $x(t)$ on $[0,T]$.
We allow (momentarily) for the possibility that $T = +\infty$.

\begin{lemma}
\label{lem:dual-structure}
\label{lem:cts_path_exists} Suppose that $x_{\ell,1}(0) > \delta$. The
continuous trajectory $x(t)$ defined in (\ref{eq:traj2}) exists uniquely
for $t\in[0,T]$ and
satisfies $x_{\ell,1}(t) \geq \delta$ for all $t\in[0,T]$. Furthermore, $x(t)$
is absolutely continuous and its derivative is given by
   \begin{equation}\label{eq:partialx_lem}
      \partial_t x_{u,i}(t) = \frac{x_{u,i}(t)+\delta}{w_u} \left(-\1_{\{(u,i)=(\ell,1)\}} - \lambda_{u,i}(t)\right)\,,
   \end{equation}
   for all $u \in V \setminus \{\root\}$ and all $i \geq 1$, and
   \begin{equation}\label{eq:lagmults_lem}
      \lambda_{u,i}(t) =
      \sum_{\substack{S \subseteq \ch(u) : \\ i \leq |S|}} \lambda_{S}(t) - \sum_{\substack{S \subseteq \ch(p(u)) : \\ (u,i) \in S}} \lambda_S(t)\,,
   \end{equation}
where
$\lambda_S(t) \geq 0$ are the Lagrangian multipliers for the constraints $\{ \sum_{i \leq |S|} x_{u,i}(t) \leq \sum_{(v,j) \in S} x_{v,j}(t) \}$. 
Also, we have that
   \begin{equation}\label{eq:compli-slack_lem}
      \lambda_S(t) > 0 \implies \sum_{i \leq |S|} x_{u,i}(t) = \sum_{(v,j) \in S} x_{v,j}(t)\,.
   \end{equation}
\end{lemma}
\begin{remark*}
   We will establish that $T < \infty$ in \pref{lem:kserver-breg}.
\end{remark*}
\begin{proof}
Since the assignment polytope $ \mathsf{A}$ is compact and convex, $\Phi$
is strongly convex and smooth, the existence and the uniqueness
of the path $x(t)$ defined in (\ref{eq:traj2}) follows
from Theorem \ref{thm:mirror_exists} with $f(t,x)=-e_{\ell,1}$. In particular, using the formula for $\Phi$, we have that
\begin{equation}\label{eq:traj3}
   \partial_t x(t) = \frac{x_{u,i}(t)+\delta}{w_u} (- e_{\ell,1} - \lambda(t))
\end{equation}
with $\lambda(t) \in N_{\mathsf{A}}(x(t))$.

To calculate $N_{\mathsf{A}}(x(t))$, we note that the constraints $\{ x_{u,i}(t) \geq 0 \}$ are redundant, as they
can be expressed by the constraints $\{ \sum_{i \leq |S|} x_{u,i} (t) \leq \sum_{(v,j) \in S} x_{v,j} (t) \}$ and $x_{\root,1} (t) = 0$ using
the sequence of singleton sets $S = {(u,i)}, {(p(u),1)}, {(p(p(u)),1)}, \cdots$.  Hence, we can ignore the constraints $\{ x_{u,i}(t) \geq 0 \}$ from the polytope.

Using the definition of the assignment polytope $\mathsf{A}$, Lemma
\ref{lem:normal_polytope} asserts that
\begin{align*}
   N_\mathsf{A}(x&(t)) \vphantom{\bigoplus} \\
= & \left\{ \ \sum_{u\in V,S \subseteq \ch(u)}\lambda_{S}(t)\left(\sum_{i\leq|S|}e_{u,i}-\sum_{(v,j)\in S}e_{v,j}\right)
+\sum_{i \geq 1}\mu_{i}(t)e_{\root,i}+\sum_{u \in V, i \geq 1}\eta_{u,i}(t)e_{u,i} \right. \\
& \vphantom{\frac{\bigoplus}{\bigoplus}} \quad\quad\text{where} \  \lambda_{S}(t), \eta_{u,i}(t) \geq 0, \  \mu_{i}(t)\in\R, \eta_{u,i}(t) \cdot (1-x_{u,i}(t)) = 0, \\ 
& \left. \quad\ \ \text{ and }\lambda_{S} (t) \cdot  \left(\sum_{i\leq|S|}x_{u,i}(t)-\sum_{(v,j)\in S}x_{v,j}(t)\right)=0\ \right\} .
\end{align*}
Rearranging the terms in (\ref{eq:traj3}) and $N_\mathsf{A}(x(t))$, and ignoring the terms for the root $\root$, we see that it only remains to show that one can take $\eta_{u,i}(t)$ (the Lagrange multiplier for $\{ x_{u,i}(t) \leq 1 \}$) to be $0$.

Let $\widetilde{\mathsf{A}}$ be the polytope just as ${\mathsf{A}}$, except
with $[0,1]$ replaced by $[0,2]$. Assume now that $x(t)$ is defined with ${\mathsf{A}}$ replaced by $\widetilde{\mathsf{A}}$.
We will show that one has $x_{u,i}(t) \leq 1$. This implies that the Lagrange multipliers for the path defined on $\widetilde{\mathsf{A}}$ are valid Lagrange multipliers for the path on ${\mathsf{A}}$, and in particular one can take $\eta_{u,i}(t) = 0$.

Toward deriving a contradiction,
let us assume that there exists a time $t>0$, $u \in V_h$, and $i \geq 1$ such that ${x}_{u,i}(t) > 1$ and $\partial_t {x}_{u,i}(t) > 0$.
We prove by induction on $h$ that this impossible.

For $h=0$ this follows from the equality constraints at the root.
Now consider $h \geq 1$ and observe that by \eqref{eq:partialx_lem} and \eqref{eq:lagmults_lem},
one must have ${\lambda}_S(t) \neq 0$ for some $S\subseteq \ch(p(u))$, which means
that $\sum_{i \leq |S|} x_{p(u),i}(t) = \sum_{(v,j) \in S} x_{v,j}(t)$.
However, the induction hypothesis implies that
for any $j \geq 1$, ${x}_{p(u),j}(t) \leq 1$, and thus the constraint corresponding to $S \setminus \{(u,i)\}$ is violated for ${x}(t)$, yielding
a a contradiction.
\end{proof}

We now prove several lemmas giving a more refined understanding of the dynamics \eqref{eq:partialx_lem}.
The reader is encouraged to skip these arguments upon a first reading.
The main technical property we need to establish is that $z(t) \in \widehat{M}$ for all times $t \in [0,T]$,
i.e., the mass per level remains constant.  This is established in \pref{lem:extended_measures}.

For $h \geq 0$,
let $V_{h}$ denote the set of vertices with a simple path to the root using $h$ edges.
Define $\cC(t) \supseteq \{S : \lambda_S(t) \neq 0\}$ to be the set of {\em active constraints:}
\[
\cC(t) \seteq \left\{S \subseteq \chi(u) : u \in V \text{ and } \sum_{i\leq\left|S\right|}x_{u,i}(t)=\sum_{(v,j)\in S}x_{v,j}(t) \right\}.
\]

\begin{lemma} \label{lem:ass0}
For any $t \geq 0$, $u \in V$, and $i \geq j \geq 1$, it holds that $x_{u,i}(t) \geq x_{u,j}(t)$.
\end{lemma}

\begin{proof}
We will show that $x_{u,i}(t)>x_{u,i+1}(t)$ implies $\partial_{t}x_{u,i}(t) \leq \partial_{t}x_{u,i+1}(t)$.
Recalling \eqref{eq:partialx_lem} and \eqref{eq:lagmults_lem}, it is enough to show that
\[
\sum_{\substack{S \subseteq \ch(p(u)) : \\ (u,i+1) \in S}} \lambda_S(t) \geq \sum_{\substack{S \subseteq \ch(p(u)) : \\ (u,i) \in S}} \lambda_S(t) ~.
\]
Let us show that if $(u,i) \in S$ and $(u,i+1) \not\in S$ then $\lambda_S(t) = 0$, yielding the desired conclusion.

Using the constraint for $S \cup \{(u,i+1)\} \setminus \{(u,i)\}$ gives
\[
\sum_{(v,j) \in S} x_{v,j}(t) > \sum_{(v,j) \in S \cup \{(u,i+1)\} \setminus \{(u,i)\}} x_{v,j}(t) \geq \sum_{i \leq |S|} x_{p(u),i}(t)\,,
\]
implying that $\lambda_S(t)=0$.
\end{proof}

\begin{lemma} \label{lem:ass}
Consider $u \in V$ and $S,S' \subseteq \chi(u)$ such that $S,S' \in \cC(t)$.  Then $S \cup S' \in \cC(t)$ as well.
\end{lemma}

\begin{proof}
   First we claim that $M_S \seteq \max_{(v,j) \in S} x_{v,j}(t)$ and $M_{S'} \seteq \max_{(v,j) \in S'} x_{v,j}(t)$ are equal.
Let $(v_*,j_*)$ denote some pair for which $x_{v,j}(t)$ is maximal among $(v,j) \in S$.
If $M_S > M_{S'}$, then for any $(v',j') \in S'$,
the constraint corresponding to $S \cup \{(v',j')\} \setminus \{(v_*,j_*)\}$ is violated because $S,S' \in \cC(t)$.
Let us denote $M \seteq M_S = M_{S'}$.

Suppose that $|S| \geq |S'|$.
The same argument shows that for any $(v,j) \in S \setminus S'$, one has $x_{v,j}(t) = M$.
Using the constraint for $S \setminus \{(v_*, j_*)\}$
and the fact that $S \in \cC(t)$ shows that $x_{u,|S|}(t) \geq M$,
and thus by \pref{lem:ass0}, one has $x_{u,|S|+m}(t) \geq M$ for any $m \geq 0$. 
This implies:
\[
\sum_{i \leq |S \cup S'|} x_{u,i}(t) \geq \sum_{i \leq |S|} x_{u,i}(t) + M \cdot |S \setminus S'| = \sum_{(v,j) \in S} x_{v,j}(t) + M \cdot |S \setminus S'| = \sum_{(v,j) \in S \cup S'} x_{v,j}(t) ~.
\]
Furthermore, since $x(t) \in \mathsf{A}$ it also holds that
$\sum_{i \leq |S \cup S'|} x_{u,i}(t)  \leq \sum_{(v,j) \in S \cup S'} x_{v,j}(t)$, demonstrating that $S \cup S' \in \cC(t)$.
\end{proof}

\begin{lemma}
\label{lem:extended_measures}
Suppose that $z(0)\in\widehat{M}$. Then $z(t)\in\widehat{M}$ for $t \geq 0$.
\end{lemma}
\begin{proof}
For $u\in V$, let $S_{u}(t)$ be the maximum (w.r.t. inclusion) active set at time $t$ in $\ch(u)$ (cf. \pref{lem:ass}).
Since $\partial_{t}x(t)\in N_{\mathsf{A}}(x(t))^{\perp}$ (\pref{lem:xt_boundary}), one has
\[
\sum_{i\leq\left|S_{u}(t)\right|}\partial_{t}x_{u,i}(t)=\sum_{(v,j)\in S_{u}(t)}\partial_{t}x_{v,j}(t) ~,
\]
which in turns gives, for any $h \geq 1$,
\begin{equation}
\sum_{u \in V_h, i \leq |S_u(t)|}\partial_{t}x_{u,i}(t)=\sum_{(v,j)\in \cup_{u \in V_h} S_u(t)}\partial_{t}x_{v,j}(t)\,.\label{eq:AC_equal}
\end{equation}

Thus to compare the derivatives of the mass at two adjacent levels,
it remains to establish that
\[\sum_{u \in V_h, i > |S_u(t)|}\partial_{t}x_{u,i}(t) \geq \sum_{(v,j) \not\in \cup_{u \in V_h} S_u(t) : v \in V_{h+1}}\partial_{t}x_{v,j}(t)\,.\]
   We show that every term in the first sum is nonnegative and every term in the second sum is nonpositive.
   In particular, since $\partial_t x_{\root,i}(t)=0$ for all $i \geq 1$, this will imply by induction that
\[
\sum_{v \in V_{h+1}, j \geq 1} \partial_{t}x_{v,j}(t) \leq \sum_{u \in V_{h}, i \geq 1} \partial_{t}x_{u,i}(t) \leq 0\,,
\]
yielding $\sum_{u \in V_h} \partial_t z_u(t) \geq 0$.
Then the proof is concluded using $z(0) \in \widehat{M}$ and $z(t) \in \mathsf{A}$ for all $t \geq 0$.

Thus it remains to show that $\partial_{t}x_{u,i}(t)\geq0$ for all $i > |S_u|$ and $u \notin \cL$, and $\partial_t x_{v,j}(t) \leq 0$ for all $(v,j) \not\in \cup_{u \in V_h} S_u(t)$. For $u=\root$, one has $\partial_{t}x_{u,i}(t)=0$, and for $u \neq r$
we have thanks to \eqref{eq:partialx_lem} and \pref{lem:dual-structure}:
\[
\partial_{t}x_{u,i}(t)= - \frac{x_{u,i}(t)+\delta}{w_{u}} \lambda_{u,i}(t) ~.
\]
Since $i>\left|S_{u}(t)\right|$ and $S_u(t)$ is the maximum active set in $\ch(u)$, it holds that $\lambda_S(t)=0$
for any $S \subseteq \ch(u)$ with $|S| \geq i$. Thus from \eqref{eq:lagmults_lem}, we see that $\lambda_{u,i}(t)\leq0$, and in turn $\partial_t x_{u,i}(t) \geq 0$. On the other hand for $(v,j) \in \ch(u)$ one has
\[
\partial_{t}x_{v,j}(t) \leq - \frac{x_{v,j}(t)+\delta}{w_{v}} \lambda_{v,j}(t) ~.
\]
Assume $(v,j) \not\in S_u(t)$. Then since $S_u(t)$ is the maximum active set in $\ch(u)$,
it holds that $\lambda_S(t)=0$
for any $S$ with $(v,j) \in S$.
Using \eqref{eq:lagmults_lem}, we see that $\lambda_{v,j}(t) \geq 0$, concluding the proof
\end{proof}

We have established that $z(t)$ is an internal measure for every $t \geq 0$, and thus $\partial_t z(t)$
is a flow.  We now we show that $\partial_t z(t)$ is a flow directed toward the request $\ell$.

\begin{lemma}
\label{lem:dual-structure2}
It holds that $\partial_t z_{u}(t)\geq0$ if $u$ is an ancestor of the request $\ell$, and $\partial_t z_{u}(t)\leq 0$ otherwise.
\end{lemma}

\begin{proof}
Since $z(t) \in \widehat{M}$ (\pref{lem:extended_measures}),
it suffices to show that for any leaf $\ell' \neq \ell$, one has $\partial_t z_{\ell'}(t) \leq 0$.
Indeed by preservation of mass at every node (i.e., $\partial_t z_u(t) = \sum_{v : (u,v) \in \vec{E}} \partial_t z_v(t)$),
this implies that $\partial_t z_u(t) \leq 0$ for any $u$ which is not an ancestor of $\ell$.
Furthermore, by preservation of mass per level (i.e., $\sum_{u \in V_h} \partial_t z_u(t) = 0$),
and the fact that there is a single ancestor of $\ell$ per level,
this also gives $\partial_t z_u(t) \geq 0$ for any ancestor $u$ of $\ell$.

Notice that $\partial_t z_{\ell'}(t) = -\frac{1}{1-\delta} \sum_{i\geq1}\partial_{t}x_{\ell',i}(t)$,
and thus it suffices to show that for any $i \geq 1$, $\partial_t x_{\ell',i}(t) \geq 0$. The latter inequality is straightforward
from \eqref{eq:partialx_lem} and \eqref{eq:lagmults_lem} since $\ch(\ell') = \emptyset$.
\end{proof}

The next lemma follows immediately from \pref{lem:dual-structure2}
since $\partial_t x_{\ell'}(t) \geq 0$ for all $\ell' \neq \ell$.

\begin{lemma}
   If $x_{\ell'}(0) \geq \delta$ for all $\ell' \in \cL$, then
   $x_{\ell'}(T) \geq \delta$ for all $\ell' \in \cL$.
\end{lemma}

Let us extend the definition of $\|\cdot\|_{\ell_1(w)}$ to $x \in \R^{\Lambda}$ by
\[
\|x\|_{\ell_1(w)} \seteq \sum_{u \in V} w_u \sum_{i \geq 1} |x_{u,i}| ~.
\]
Observe that
\[
   \sup_{x \in \mathsf{A}} \|\nabla \Phi(x)\|_{\ell_\infty(1/w)} \leq O(\log \tfrac{1}{\delta})\,.
\]
This yields the following.

\begin{lemma} \label{lem:optmoves}
   It holds that for every $x \in \mathsf{A}$ and $\{y(t)\} \subseteq \mathsf{A}$ differentiable:
   \[
      \partial_t \hat{D}_{\Phi}(y(t); x) \leq \|\partial_t y(t)\|_{\ell_1(w)} O(\log \tfrac{1}{\delta})\,.
   \]
\end{lemma}

The next lemma is an immediate consequence of \eqref{eq:lower-breg}.

\begin{lemma}\label{lem:kserver-breg}
   If $y \in \mathsf{A}$ satisfies $y_{\ell}=0$, then
   \[
      \partial_t \hat{D}_{\Phi}(y; x(t)) \leq
         - x_{\ell,1}(t)\,.
   \]
In particular, we have that $T < \infty$ and hence $x_{\ell,1}(T) = \delta$.
\end{lemma}

\begin{proof}
The first conclusion follows from \eqref{eq:lower-breg} using $y_{\ell}=0$ and $f(t,x(t)) = -e_{\ell,1}$.
Since the divergence is nonnegative and it is decreasing with rate at least $\delta$
whenever $x_{\ell,1}(t) \geq \delta$, the trajectory ends in finite time.
\end{proof}

\subsection{The weighted depth potential}
\label{sec:depth}

Let us first define an auxiliary potential function $\Psi_t$.
We relate it to the dynamics, and then
present applications in \pref{sec:depth}--\pref{sec:weighted-depth}, culminating
in the assertion that our algorithm is $O((\log k)^2)$-competitive.

Consider a differentiable function $\Psi(t)$ such that
\begin{equation}\label{eq:Psidiff}
   \partial_t \Psi(t) = \sum_{u \in V \setminus \{\root\}} w_u \left(\Delta_u(t)+\Delta_{p(u)}(t)\right) \sum_{i \geq 1} \partial_t x_{u,i}(t)
\end{equation}
for some functions $\{\Delta_u(t) \geq 0 : u \in V\}$ satisfying $\Delta_u(t) \leq \Delta_v(t)$
for all $(u,v) \in \vec{E}$ as well as $\Delta_{\root}(t) \equiv 0$.
Note that an important special case is simply when
\[
   \Psi(t) = \sum_{u \in V \setminus \{\root\}} w_u (\Delta_u+\Delta_{p(u)}) \sum_{i \geq 1} x_{u,i}(t)
\]
where $\{ \Delta_u :u \in V\}$ are independent of $t$.
For an edge $(u,v) \in \vec{E}$, define
\begin{equation}\label{eq:qdef}
q_t(v) \seteq \Delta_{v}(t) - \Delta_u(t)\,.
\end{equation}

\begin{lemma}\label{lem:depth}
   The following holds for almost every $t \in [0,T]$:
   \[
      \left\|\partial_t x(t)\right\|_{\ell_1(q_t w)} \leq 3 \Delta_{\ell}(t) (x_{\ell}(t)+\delta) + \partial_t \Psi(t)\,.
   \]
\end{lemma}

\begin{proof}
   Note that from Lemma \ref{lem:dual-structure}, for every $u \in V \setminus \{\root\}$, we have:
   \begin{equation}\label{eq:partialx}
      \partial_t x_{u,i}(t) = \frac{x_{u,i}(t)+\delta}{w_u} \left(-\1_{\{(u,i)=(\ell,1)\}} - \lambda_{u,i}(t)\right)\,,
   \end{equation}
   where
   \begin{equation}\label{eq:lagmults}
      \lambda_{u,i}(t) =
      \sum_{\substack{S \subseteq \ch(u) : \\ i \leq |S|}} \lambda_{S}(t) - \sum_{\substack{S \subseteq \ch(p(u)) : \\ (u,i) \in S}} \lambda_S(t)\,.
   \end{equation}
   Recalling that $\Delta_{\root}(t) = 0$ for all $t \in [0,T]$, we calculate:
   \begin{align*}
   \partial_t \Psi(t)\ + &\ 2 \Delta_{\ell}(t) (x_{\ell}(t)+\delta) \\
                         &\geq \partial_t \Psi(t)\ + \left(\Delta_{\ell}(t) + \Delta_{p(\ell)}(t)\right) (x_{\ell}(t)+\delta) \\
                         &= \left(\Delta_{\ell}(t) + \Delta_{p(\ell)}(t)\right) (x_{\ell}(t)+\delta) + \sum_{u \in V \setminus \{\root\}} w_u \left(\Delta_u(t)+\Delta_{p(u)}(t)\right) \sum_{i \geq 1} \partial_t x_{u,i}(t)
      \nonumber \\
   &=
   \sum_{u \in V \setminus \{\root\}} \sum_{S \subseteq \ch(u)} \lambda_S(t)
   \left(\sum_{(v,j) \in S} \left(\Delta_v(t)+\Delta_{p(v)}(t)\right) (x_{v,j}(t)+\delta)-\left(\Delta_u(t)+\Delta_{p(u)}(t)\right) \sum_{i \leq |S|} (x_{u,i}(t)+\delta) \right) \nonumber \\
& \qquad + \sum_{S \subseteq \ch(\root)} \lambda_S(t)
  \sum_{(v,j) \in S} \Delta_v(t) (x_{v,j}(t)+\delta) \nonumber \\
   &=
   \sum_{u \in V \setminus \{\root\}} \sum_{S \subseteq \ch(u)} \lambda_S(t)
   \left(\sum_{(v,j) \in S} \left[\left(\Delta_v(t)-\Delta_u(t)\right) + \left(\Delta_{p(v)}(t)-\Delta_{p(u)}(t)\right)\right] (x_{v,j}(t)+\delta)\right) \nonumber \\
& \qquad + \sum_{S \subseteq \ch(\root)} \lambda_S(t)
  \sum_{(v,j) \in S} \Delta_v(t) (x_{v,j}(t)+\delta) \nonumber \\
   &=
   \sum_{u \in V \setminus \{\root\}} \sum_{S \subseteq \ch(u)} \lambda_S(t)
   \left(\sum_{(v,j) \in S} \left(q_t(v)+q_t(u)\right) (x_{v,j}(t)+\delta)\right) \nonumber \\
& \qquad
+ \sum_{S \subseteq \ch(\root)} \lambda_S(t)
  \sum_{(v,j) \in S} \Delta_v(t) (x_{v,j}(t)+\delta) \,,
      \end{align*}
   where in the penultimate equality we have used that
   for $S \subseteq \ch(u)$,
   \begin{equation}\label{eq:compli-slack}
      \lambda_S(t) > 0 \implies \sum_{i \leq |S|} x_{u,i}(t) = \sum_{(v,j) \in S} x_{v,j}(t)\,.
   \end{equation}

   On the other hand, the $q_t\cdot w$-movement cost is equal to
   \[
      \sum_{u \in V \setminus \{\root\}} \sum_{i \geq 1} q_t(u) w_u \left|\partial_t x_{u,i}(t)\right|  =
      \sum_{u \in V \setminus \{\root\}} \sum_{i \geq 1} q_t(u) w_u \left(\partial_t x_{u,i}(t)\right)_+
      - \sum_{u \in V \setminus \{\root\}} q_t(u) w_u \left(\partial_t x_{u,i}(t)\right)_{-}\,.
   \]
   Using \eqref{eq:partialx} and \eqref{eq:lagmults} gives
\[
      \sum_{u \in V \setminus \{\root\}} \sum_{i \geq 1} w_u q_t(u) \left(\partial_t x_{u,i}(t)\right)_+
      \leq
      \sum_{u \in V} \sum_{S \subseteq \ch(u)} \lambda_S(t) \sum_{(v,j) \in S} q_t(v) (x_{v,j}(t)+\delta)
\]
   and
   \[
      -\sum_{u \in V \setminus \{\root\}} \sum_{i \geq 1} w_u q_t(u) \left(\partial_t x_{u,i}(t)\right)_-
      \leq q_t(\ell) (x_{\ell}(t)+\delta) + 
      \sum_{u \in V \setminus \{\root\}} \sum_{S \subseteq \ch(u)} \lambda_S(t) q_t(u) \sum_{i \leq |S|} \left(x_{u,i}(t)+\delta\right)
   \]
   Using $q_t(\ell) \leq \Delta_t(\ell)$ and \eqref{eq:compli-slack} again (as well as $q_t(v) = \Delta_v(t)$ for $(v,j) \in S \subseteq \ch(\root)$), we have
   \begin{align*}
      \sum_{u \in V \setminus \{\root\}} \sum_{i \geq 1} q_t(u) w_u \left|\partial_t x_{u,i}(t)\right|
      &\leq \Delta_t(\ell) (x_{\ell}(t)+\delta) + 
   \sum_{u \in V \setminus \{\root\}} \sum_{S \subseteq \ch(u)} \lambda_S(t) \sum_{(v,j) \in S} \left(q_t(u)+q_t(v)\right) (x_{v,j}(t)+\delta) \\
& \qquad + \sum_{S \subseteq \ch(\root)} \lambda_S(t)
  \sum_{(v,j) \in S} \Delta_v(t) (x_{v,j}(t)+\delta) \,,
\end{align*}
yielding the desired result.
\end{proof}

Note that $\|\partial_t x(t)\|_{\ell_1(q_t w)} = (1-\delta) \|\partial_t z(t)\|_{\ell_1(q_t w)}$.
Thus combining \pref{lem:depth} with \pref{lem:kserver-breg} and using $x_{\ell}(t) \geq \delta$ for $t \in [0,T]$
yields the following.

\begin{corollary}\label{cor:depth}
   For almost every $t \in [0,T]$, if $y \in \mathsf{A}$ satisfies
   $y_{\ell}=0$, then
   \[
      (1-\delta) \left\|\partial_t z(t)\right\|_{\ell_1(q_t w)}
      \leq \partial_t \Psi(t) + 6 \Delta_{\ell}(t) \partial_t \hat{D}_{\Phi}(y; x(t))\,.
   \]
\end{corollary}

For a function $f : V \to \R_+$, define $\widehat{f} : V \setminus \cL \to \R_+$
by
\[
   \widehat{f}(u) \seteq \min \left\{ f(v)+f(v') : v \neq v', (u,v),(u,v') \in \vec{E} \right\}\,.
\]
For concreteness, let us define $\widehat{f}(u)\seteq 0$ if $u$ has only one child.

\begin{lemma}\label{lem:depth2}
   For almost every $t \in [0,T]$, the following holds.
   Suppose that $(\cT,w)$ is a $\tau$-adic HST for $\tau \geq 2$,
   and there is some $c > 0$ such that
   \[
      \widehat{q}_t(v) \geq c \qquad \forall v \in V \setminus \cL\,.
   \]
   If $y \in \mathsf{A}$
   satisfies $y_{\ell}=0$, then
   \[
      \frac{c(1-\delta)}{4} \left\|\partial_t z(t)\right\|_{\ell_1(w)} \leq \partial_t \Psi(t)
      + 6 \Delta_{\ell}(t) \partial_t \hat{D}_{\Phi}(y; x(t))\,.
   \]
\end{lemma}

\begin{proof}
   From \pref{lem:extended_measures}, it holds that $z(t)$ is an internal measure for all $t \in [0,T]$, 
   and moreover $\partial_t z(t)$ is a flow towards $\ell$ (cf. \pref{lem:dual-structure2}).
   Therefore we can decompose
   \[
      \partial_t z(t) = \sum_{\ell' \in \cL} y^{(\ell')}(t)\,,
   \]
   where $y^{(\ell')}(t)$ is a flow on the unique $\ell'$-$\ell$ path in $\cT$.

   Let us use $|y^{(\ell')}(t)|$ to denote the magnitude of the corresponding flow.
   Since $\partial_t z(t)$ is a flow towards $\ell$, we have
   \begin{align}
      \left\|\partial_t z(t)\right\|_{\ell_1(q_t w)} &\geq \frac{1}{\tau} \sum_{\ell' \in \cL \setminus \{\ell\}} \widehat{q}_t(\lca(\ell,\ell')) w_{\lca(\ell,\ell')}
      \left|y^{(\ell')}(t)\right| \label{eq:qtw} \\
      &\geq \frac{c}{\tau} \sum_{\ell' \in \cL \setminus \{\ell\}} w_{\lca(\ell,\ell')}
      \left|y^{(\ell')}(t)\right|
      \geq \frac{c}{4} \left\|\partial_t z(t)\right\|_{\ell_1(w)},\nonumber
   \end{align}
   where the first and last inequality use the fact that $w$ is $\tau$-adic.
\end{proof}

\subsubsection{Combinatorial depth for general trees}
\label{sec:depth}

Let $\dist_{\cT}$ denote the unweighted shortest-path metric on $\cT$.
Define $\Delta_u(t) \seteq \Delta_u = \dist_{\cT}(\root,u)$, and
\[
   \Psi(t) \seteq \sum_{u \in V \setminus \{\root\}} w_u \left(\Delta_u(t)+\Delta_{p(u)}(t)\right) \sum_{i \geq 1} x_{u,i}(t)\,,
\]
Note that $q_t$ (cf. \eqref{eq:qdef}) satisfies $q_t \equiv 1$ for all $t \in [0,T]$.
Therefore applying \pref{cor:depth} yields
\[
   (1-\delta) \left\|\partial_t z(t)\right\|_{\ell_1(w)} \leq \partial_t \Psi(t) + 6 \dist_{\cT}(\root,\ell) \cdot \partial_t \hat{D}_{\Phi}(y;x(t))\,,
\]
which in turn gives the following result (recall also \pref{lem:optmoves}).

\begin{corollary} \label{cor:generaltree}
For any tree metric with combinatorial depth $D$, 
there is an $O(D \log k)$-competitive {\em fractional} $k$-server algorithm.
\end{corollary}

Note that, as opposed to the situation for HSTs,
it is not known how to round online a fractional $k$-server algorithm on
a tree to a random integral algorithm while losing only an $O(1)$ factor
in the competitive ratio.

\subsubsection{Cardinality for an HST}
\label{sec:logn}

Assume now that $(\cT,w)$ is a $\tau$-adic HST for some $\tau \geq 2$.
Recall that $N_u$ is
number of leaves in the subtree rooted at $u$ and
define $\Delta_u(t) \seteq \Delta_u = \log \left(\frac{n}{N_u}\right)$.
Define:
\[
   \Psi(t) \seteq \sum_{u \in V \setminus \{\root\}} w_u \left(\Delta_u(t)+\Delta_{p(u)}(t)\right) \sum_{i \geq 1} x_{u,i}(t) \,.
\]
Define $q = q_t$ as in \eqref{eq:qdef}.
Then for every $u \in V$:
\[
   q(u) = \log \frac{N_{p(u)}}{N_u}\,.
\]
In particular, for any two children $v,v'$ of $u$:
\[
   \widehat{q}(u) \geq \log \frac{N_u}{N_v} + \log \frac{N_u}{N_{v'}} \geq \log 2\,.
\]

Applying \pref{lem:depth2} with $c \seteq \log 2$ yields
\[
   \frac{(1-\delta) \log 2}{4}\|\partial_t z(t)\|_{\ell_1(w)} \leq \partial_t \Psi(t) + 6\log(n)\cdot \partial_t \hat{D}_{\Phi}(y; x(t))\,.
\]
Combined with \pref{lem:optmoves}, this gives the following consequence.

\begin{corollary}
   If $(\cT,w)$ is a $\tau$-adic HST for some $\tau \geq 2$, then there is an $O(\log (k) \log (n))$-competitive
   online fractional $k$-server algorithm on $(\cL,\dist_w)$.
\end{corollary}

The preceding construction motivates our approach to obtaining
an $O((\log k)^2)$ competitive ratio:  Try to replace $N_u$
by the fractional server mass in the subtree beneath $u$.

\subsubsection{Fractional server-weighted depth}
\label{sec:weighted-depth}

Finally, let us establish the $O((\log k)^2)$ bound.
Suppose now that $(\cT,w)$ is a $\tau$-adic HST.
Define:
\[
   \Psi(t) \seteq \sum_{u \in V \setminus \{\root\}} w_u \left[\left(z_u(t)+\left(1+\tau^{-1} \1_{\{u \notin \cL\}}\right) \e\right)
                     \log (z_u(t)+\e) + z_u(t) \log (z_{p(u)}(t)+\e)\right]\,.
\]
Consider some node $u \in V \setminus \{\root\}$, and the terms
in $\partial_t \Psi(t)$ corresponding to $\partial_t z_u(t)$:
\begin{align*}
   w_u \partial_t z_u&(t) \left(1+\log\left(z_u(t)+\e\right) + \log\left(z_{p(u)}(t)+\e\right) + \frac{\tau^{-1} \e \1_{\{u \notin \cL\}}}{z_u(t)+\e}
   + \tau^{-1} \1_{\{u \notin \cL\}} \sum_{v : (u,v) \in \vec{E}} \frac{z_v(t)}{z_u(t)+\e}\right) \\
   &= \vphantom{\bigoplus} w_u \partial_t z_u(t)
         \left(\vphantom{\bigoplus}1+\log\left(z_u(t)+\e\right) + \log\left(z_{p(u)}(t)+\e\right) + \1_{\{u \notin \cL\}} \tau^{-1}\right)\,,
\end{align*}
where in the last equality we used that $z(t) \in \widehat{M}$.

Since $z(t) \in \widehat{M}$ and $(\cT,w)$ is a $\tau$-adic HST,
it holds that for every $j \in \Z$:
\[
   \sum_{\substack{u \in V : \\ w_u = \tau^{j}}} \partial_t z_u(t) = 0\,,
\]
and therefore we conclude that
\begin{align*}
   \partial_t \Psi(t) &= \sum_{u \in V \setminus \{\root\}} 
      w_u\partial_t z_u(t) \left[\log \left(z_u(t)+\e\right) + \log \left(z_{p(u)}(t)+\e\right)\right] \\
      &= - \sum_{u \in V\setminus \{\root\}} 
      w_u\partial_t z_u(t) \left[\log \frac{k+2\e}{z_u(t)+\e} + \log \frac{k+2\e}{z_{p(u)}(t)+\e}\right] \\
      &=
      \frac{1}{1-\delta} \sum_{u \in V \setminus \{\root\}}
      w_u \left[\log \frac{k+2\e}{z_u(t)+\e} + \log \frac{k+2\e}{z_{p(u)}(t)+\e}\right] \sum_{i \geq 1} \partial_t x_{u,i}(t)\,.
\end{align*}

Therefore \eqref{eq:Psidiff} holds with
\[
   \Delta_u(t) = \frac{1}{1-\delta} \log \frac{k+2\e}{z_u(t)+\e}\,,
\]
and in this case:
\[
   q_t(u) = \frac{1}{1-\delta} \log \frac{z_{p(u)}(t)+\e}{z_u(t)+\e}\,.
\]
Now we can prove the main technical theorem of this section.

\begin{theorem}\label{thm:log2k}
   The trajectory $\sigma(z(t))$ for $t \in [0,T]$ is an internal supermeasure of mass $k$ 
   that services the request at $\ell \in \cL$.  Moreover, if $y \in \mathsf{A}$ satisfies $y_{\ell}=0$, then for almost every $t \in [0,T]$:
   \[
      \frac{1-\e}{4} \log \tfrac43 \left\|\partial_t \sigma(z(t))\right\|_{\ell_1(w)} \leq (1-\delta) \partial_t \Psi(t) + 6 \log(2+k/\e)\cdot
      \partial_t
      \hat{D}_{\Phi}(y;x(t))\,.
   \]
\end{theorem}

\begin{proof}
   First, note that if $u \in V \setminus \cL$ and
   $z_u(t) \geq \e$, then for any children $v,v'$ of $u$:
   \[
      (1-\delta) \widehat{q}_t(u) \geq \log \frac{z_u(t)+\e}{z_v(t)+\e} + \log \frac{z_u(t)+\e}{z_{v'}(t)+\e} \geq \log \frac43\,,
   \]
   and therefore
   \begin{equation}\label{eq:qhat}
      \widehat{q}_t(u) \geq \frac{\log \frac43}{1-\delta}\1_{\{z_u(t) \geq \e\}} \,.
   \end{equation}
   Let $y^{(\ell')}(t)$ be as in the proof of \pref{lem:depth2}.
   Partition $\cL \setminus \{\ell\} = \bigcup_{v \in V \setminus \cL} \cL_v$, where $\cL_v$ is the set of leaves $\ell'$ with $v=\lca(\ell,\ell')$,
   and define
   \[
      y^{(v)}(t) \seteq \sum_{\ell' \in \cL_v} y^{(\ell')}(t)\,.
   \]
   Note that by \pref{lem:dual-structure2}, $\partial_t z(t)$ is a flow towards $\ell$,
   and thus there are no cancellations in the preceding sum.

   Now use inequality \eqref{eq:qtw} and \eqref{eq:qhat} to write
   \begin{align*}
      \left\|\partial_t z(t)\right\|_{\ell_1(q_t w)} &\geq 
      \frac{\log \frac43}{(1-\delta)\tau}
      \sum_{\ell' \in \cL\setminus \{\ell\}} w_{\lca(\ell,\ell')} \1_{\left\{z_{\lca(\ell,\ell')}(t) \geq \e\right\}} \left|y^{(\ell')}(t)\right|\\
      &=
      \frac{\log \frac43}{(1-\delta) \tau}
      \sum_{v \in V \setminus \cL} w_v \1_{\{z_v(t) \geq \e\}} \left|y^{(v)}(t)\right| \\
      &\geq
      \frac{1-\e}{1-\delta} \log \tfrac43
      \left\|\partial_t \sigma(z(t))\right\|_{\ell_1(w)}\,,
   \end{align*}
   where in the final line we have used the fact that $\sigma$ is $\frac{1}{1-\e}$-Lipschitz and
   $\sigma(z_v(t))=0$ when $z_v(t) < \e$.
   Combined with \pref{cor:depth}, this yields the desired result, noting that
   for any leaf $\ell \in \cL$:
   \[
      \Delta_{\ell}(t) \leq \frac{1}{1-\delta} \log \frac{k+2\e}{\e}\,.\qedhere
   \]
\end{proof}

Using \pref{lem:optmoves} and \pref{lem:internal}, this yields an $O((\log k)^2)$-competitive online
fractional $k$-server algorithm for any HST metric.
(It is not difficult to see that every HST metric embeds with $O(1)$ distortion
into the metric of a $2$-adic HST.)
 
\section{Dynamic HST embeddings}
\label{sec:dynamic}

\newcommand{\cc}{\cC_X}
\newcommand{\bcc}{\bar{\cC}_X}

Consider a discrete metric space $(X,d)$.
Denote the {\em aspect ratio of $(X,d)$} by
\[
   \cA_X \seteq \frac{\max_{x,y \in X} d(x,y)}{\min_{x \neq y \in X} d(x,y)}\,.
\]

\begin{theorem}
   If there is an $\alpha$-competitive algorithm for $k$-server on HSTs, then there is an
   $O(\alpha \log(\cA_X) \log k)$-competitive algorithm on any metric space $(X,d)$.
\end{theorem}

We use $\cR(X) \seteq X^{\N}$ to denote
the space of request sequences.
For $\bm{\sigma} \in \cR(X)$ and $s \leq t$, denote $\bm{\sigma}_{[s,t]} \seteq \langle\sigma_s, \sigma_{s+1}, \ldots, \sigma_t\rangle$.
We will consider sequences of random variables that are implicitly functions of $\bm{\sigma} \in \cR(X)$.
Say that such a sequence $\bm{Z} = \llangle Z_t : t \geq 0\rrangle$ is {\em adapted}
if $Z_t$ is a function of $\bm{\sigma}_{[1,t]}$ for every $t \geq 1$.

This allows one to encode state that depends on the underlying request sequence $\bm{\sigma}$
in a time-dependent way.
For instance, to count the number of requests that fall into a subset $X' \subseteq X$,
one would set $\bm{Z}=0$.  Then given that $\sigma_t$ arrives at time $t$,
we would update $\bm{Z} \seteq \bm{Z} + \1_{X'}(\sigma_t)$.
The meaning is: $Z_0=0$ and $Z_t \seteq Z_{t-1} + \1_{X'}(\sigma_t)$ for all $t \geq 1$.

An {\em online algorithm for $k$-server on $X$} is an adapted sequence $\bm{A} = \langle A_1, A_2,\ldots\rangle$
where for every $t \geq 1$: $A_t \in X^k$ and $\sigma_t \in  \{ (A_t)_1,\ldots,(A_t)_k\}$.
For a function $f$ with domain $X$, write $f^{\otimes k}$ for the function with domain $X^k$
given by $f^{\otimes k}(x_1,\ldots,x_k) \seteq (f(x_1),\ldots,f(x_k))$.

For an algorithm $\bm{A}$ and a request sequence $\bm{\sigma}$, we write $\cost_X(\bm{A}(\bm{\sigma}))$
for the total movement cost incurred in servicing $\bm{\sigma}$.  We denote by $\opt^X : \cR(X) \to (X^k)^{\N}$
an optimal {\em offline} algorithm and $\cost^*_X(\bm{\sigma}) \seteq \cost_X(\opt^X(\bm{\sigma}))$
the optimal offline movement cost.

\subsection{Hierarchical partitions and canonical HSTs}

Let us suppose that $\diam(X,d)=1$ and $\cA_X < \infty$.
Let $\tau \seteq 4$ be a scale parameter, and let $M \in \N$ denote the smallest number
for which $\tau^{-M} < d(x,y)$ for all $x \neq y \in X$.

A sequence of subsets $\xi = (\xi_0,\xi_1,\ldots,\xi_{\ell})$ of $X$ for $0 \leq \ell \leq M$ is a {\em chain} if 
$\xi_0=X$ and \[\xi_0 \supseteq \xi_1 \supseteq \cdots \supseteq \xi_{\ell}\,.\]
Define the {\em length} of such a chain by $\len(\xi)\seteq \ell$.
Denote $\min(\xi) \seteq \xi_{\len(\xi)}$.
A chain is {\em complete} if it has length $M$ and $|\xi_M|=1$.
Let $\cc$ denote the set of chains in $X$ and let $\bcc$ denote
the set of complete chains.

Define a rooted tree structure on $\cC_X$ as follows.  The root of $\cc$ is $X$.
For two chains $\xi,\xi' \in \cc$:
$\xi'$ is a child of $\xi$ if and only if $\xi$ is a prefix of $\xi'$ and $\len(\xi')=\len(\xi)+1$.
For $\xi,\xi' \in \cC_X$, let $\lca(\xi,\xi') \in \cC_X$ denote their least common ancestor.
Define a $\tau$-HST metric on $\bcc$
by
\[
   \hst_{\tau}(\xi,\xi') \seteq \tau^{-\len(\lca(\xi,\xi'))} \qquad \xi \neq \xi'\,.
\]

\medskip
\noindent
{\bf Embedding into complete chains.}
For a partition $P$ of $X$ we write $P(x)$ for the unique set $S \in P$ containing $x$.
A {\em $\tau$-stack $\cP$ of $X$} is a sequence $\cP = (P^0,P^1,\ldots,P^M)$
of partitions of $X$ such that:  $P^0=\{X\}$ and for all $j=1,2,\ldots,M$, it holds that
$$S \in P^j \implies \diam(S) \leq \tau^{-j} \, .$$ 
Note that $P^M(x)=\{x\}$ because $\diam(S) \leq \tau^{-M}$ implies $|S|\leq 1$.
Every $\tau$-stack $\cP$ induces a canonical mapping $F_{\cP} : X \to \bcc$
into the set of complete chains on $X$ as follows. First define the {\em forced refinement} $\hat{\cP} = \left(\hat{P}^0, \hat{P}^1, \ldots, \hat{P}^M\right)$ 
inductively by
$\hat{P}^0 \seteq P^0$ and $\hat{P}^{j} \seteq \left\{ S \cap S' : S \in P^j, S' \in P^{j-1} \right\}$ for $j=1,2,\ldots,M$. Next define $F_{\cP}$ by
\[
   F_{\cP}(x) \seteq \left(\hat{P}^0(x), \hat{P}^1(x), \ldots, \hat{P}^M(x)\right) \, .
\]

The following two lemmas will help to estimate the distortion of the embedding $F_{\cP}$.

\begin{lemma}\label{lem:non-contract}
   For any $\tau$-stack $\cP$, the map $F_{\cP} : X \to (\bcc,\hst_{\tau})$ is non-contracting.
\end{lemma}

\begin{proof}
   Consider $x,y \in X$.  If $\hst_{\tau}(F_{\cP}(x), F_{\cP}(y)) = \tau^{-\ell}$, then $F_{\cP}(x)$ and $F_{\cP}(y)$
   share a common prefix of length $\ell$, and therefore $\cP^{\ell}(x)=\cP^{\ell}(y)$.
   Now property (ii) of a $\tau$-stack guarantees that $d(x,y) \leq \diam(\cP^{\ell}(x)) \leq \tau^{-\ell}$.
\end{proof}

\begin{lemma}\label{lem:prestack}
   For any $\tau$-stack $\cP=\left(P^0, P^1, \ldots, P^M\right)$, it holds that
   \[
      \hst_{\tau}\left(F_{\cP}(x),F_{\cP}(y)\right) \leq \tau \sum_{j \geq 1} \tau^{-j} \1_{P_j(x) \neq P_j(y)} \qquad \forall x,y \in X\,.
   \]
\end{lemma}

\begin{proof}
   Consider $x,y \in X$ and suppose that $\hst_{\tau}\left(F_{\cP}(x),F_{\cP}(y)\right) = \tau^{-\ell}$ for some $\ell < M$.
   It is straightforward to check that $\ell+1 = \min \{ j : P^j(x) \neq P^j(y) \}$.
\end{proof}

\newcommand{\inver}{\mathsf{in}}
\newcommand{\fin}{F_{\inver}}

Next we observe that there is a universal inverse to those embeddings. Define the mapping $F_{\inver} : \bcc \to X$ by $F_{\inver}(\xi) \seteq \min(\xi)$. One has {\em for any} $\tau$-stack $\cP$:
\begin{equation}\label{eq:inverse}
   F_{\inver} \circ F_{\cP} = \id_X\,.
\end{equation}

\subsection{The HST reduction}

Let $\bm{A}^{\cC}$ denote an $\alpha$-competitive
$k$-server algorithm for the metric space $(\bcc, \hst_{\tau})$
over some probability space $\Omega_{\cC}$.
Suppose that $\bm{\sigma} = \langle \sigma_1, \sigma_2, \ldots\rangle$ is a request sequence for $X$,
and we have an adapted sequence $\bm{\cP} = \langle \cP_1, \cP_2, \ldots\rangle$ of $\tau$-stacks of $X$
over an independent probability space $\Omega_X$.

This yields a mapping $F_{\bm{\cP}} : \cR(X) \to \cR(\cc)$ given by
\[
   F_{\bm{\cP}}(\sigma) \seteq \left\langle F_{\cP_1}(\sigma_1), F_{\cP_2}(\sigma_2), \ldots\right\rangle\,.
\]
From this one derives
a $k$-server algorithm for $X$:
\[
   \bm{A}^X \seteq \fin^{\otimes k} \circ \bm{A}^{\cC} \circ F_{\bm{\cP}}\,,
\]
Note that $\bm{A}^X$ is a valid $k$-server algorithm precisely because of \eqref{eq:inverse}.

Moreover, because of \pref{lem:non-contract}, the inverse map $\fin$ is non-expanding, and thus:
\begin{equation}\label{eq:red1}
  \cost_X(\bm{A}^{X}(\bm{\sigma})) \leq \E_{\Omega_{\cC}} \left[\cost_{\hst_{\tau}}\!\left((\bm{A}^{\cC} \circ F_{\bm{\cP}})(\bm{\sigma})\right)\right]
   \leq
   O_{X,k}(1) + \alpha \cost^*_{\hst_{\tau}}(F_{\bm{\cP}}(\bm{\sigma})) \, .
\end{equation}
Thus our goal becomes clear:  We would like to choose $\bm{\cP}$ so that
\begin{equation}\label{eq:red2}
   \E_{\Omega_X} \left[\cost_{\hst_{\tau}}\!\left((F_{\bm{\cP}}^{\otimes k} \circ \opt^X)(\bm{\sigma})\right)\right] \leq O_{X,k}(1) + \beta\, \cost^*_X(\bm{\sigma}) \qquad \forall \sigma \in \cR(X)\,.
\end{equation}
Indeed since $F_{\bm{\cP}}^{\otimes k} \circ \opt^X$ services the request sequence $F_{\bm{\cP}}(\bm{\sigma})$,
one has $\cost^*_{\hst_{\tau}}(F_{\bm{\cP}}(\bm{\sigma})) \leq \cost_{\hst_{\tau}}\!\left((F_{\bm{\cP}}^{\otimes k} \circ \opt^X)(\bm{\sigma})\right)$, and thus \eqref{eq:red2}
in conjunction with \eqref{eq:red1} show that $\bm{A}^X$ is an $\alpha \beta$-competitive algorithm for $X$.

In essence, \eqref{eq:red2} asks that the embedding $F_{\bm{\cP}}$ has $\beta$-distortion {\em on the subset of $X$ that currently matters}. Focusing on such a subset is the reason why one could
hope to the usual $\Omega(\log n)$ lower bound by something depending on $k$ and $\cA_X$.

\subsection{A dynamic embedding}

The algorithm will produce an adapted sequence $\bm{\cP} = \langle \cP_1, \cP_2, \ldots\rangle$ of $\tau$-stacks verifying \eqref{eq:red2} with $\beta \leq O(M \log k)$. In fact it is slightly easier for the algorithm's description to allow $P_t^j$ to be a partial partition, i.e.,
 simply a collection of pairwise disjoint subsets of $X$. In the case it is understood that the embeddings $F_{\cP}$ use the {\em completion} of any such partial partition $P$, that is all the elements from $X \setminus [P]$ (we denote $[P] \seteq \cup_{S \in P} S$) are added as singletons to form a complete partition.

\paragraph{The embedding algorithm}
For $j=1,\ldots,M$, we will maintain an adapted set $\bm{N}^j \subseteq X$
of {\em level-$j$ centers} as well as an adapted mapping $\bm{R}^j : \bm{N}^j \to \R_+$
of radii. These will be used to construct our adapted $\tau$-stack
$\bm{\cP} = \left(\bm{P}^0, \bm{P}^1, \ldots, \bm{P}^M\right)$.

For every $j \in \Z$, consider the probability distribution $\mu_j$ with density:
\[
   d\mu_j(r) \seteq \frac{k \tau^j \log k}{k-1} \exp\left(-r \tau^j \log k\right) \1_{[0,\tau^{-j}]}(r)\,.
\]
This is simply a truncated exponential distribution, as employed by Bartal \cite{Bar96}.

\smallskip
\noindent
Initially, $\bm{N}^j_0 = \bm{P}_0^j = \emptyset$ for $j=1,\ldots,M$.
Upon receiving request $\sigma_t \in X$, we proceed as follows:
\begin{quote}
For $j=1,2,\ldots,M$:
      \begin{enumerate}
         \item If $|\bm{N}^j| \geq 2k$, then
            
            [{\sf level-$j$ reset}]:
            
            $\qquad$ For $i=j,j+1,\ldots,M$,
               set $\bm{N}^j \seteq \emptyset$ and $\bm{P}^j \seteq \emptyset$.
         \item If $d(\sigma_t,\bm{N}^j) > \tau^{-j-1}$, then

            [{\sf level-$j$ insertion}]:
            \begin{enumerate}
               \item $\bm{N}^j \seteq \bm{N}^{j} \cup \{\sigma_t\}$
               \item $\bm{R}^j(\sigma_t) \seteq \tau^{-j-1}+\hat{R}^j_t$, where $\hat{R}^j_t$ is sampled independently according to the distribution $\mu_{j+1}$.
               \item $\bm{P}^j \seteq \bm{P}^j \cup \left\{B(\sigma_t, \bm{R}^j(\sigma_t)) \setminus [\bm{P}^j]\right\}$.
            \end{enumerate}
       \end{enumerate}
\end{quote}

\paragraph{Analysis}
Let us now prove that for the stack $\bm{\cP}$ generated in this way,
\eqref{eq:red2} holds for $\beta \leq O(M \log k)$.
We need a few preliminary results.

\begin{lemma}\label{lem:proper}
   $\cP_t$ is a $\tau$-stack for every $t \geq 1$.
\end{lemma}

\begin{proof}
   This follows from the fact that the distribution $\mu_{j+1}$ is supported on the interval $[0,\tau^{-j-1}]$ and thus
   every set in $P_t^j$ is contained in a ball of radius $2 \tau^{-j-1}$, which is a set of diameter at most
   $4 \tau^{-j-1} \leq \tau^{-j}$ for $\tau \geq 4$.
\end{proof}

We defer the proof of the next lemma to the end of this section.
\begin{lemma}\label{lem:stretch}
   For all $x \in X$ and $t \geq 1$, it holds that
   \[
      \E_{\Omega_X} \left[\hst_{\tau}(F_{\bm{\cP}_t}(x), F_{\bm{\cP}_t}(\sigma_t))\right] \leq O(M \log k)\, d(x,\sigma_t)\,.
   \]
\end{lemma}

\begin{lemma}\label{lem:resets}
   For all $t \geq 1$ and $j=1,2,\ldots,M$:
   If $K_{j,t}$ denotes the number of level-$j$ resets up to time $t$, then
   \[
      \cost_X^*(\sigma_{[1,t]}) \geq k \tau^{-j-1} \cdot K_{j,t}\,.
   \]
\end{lemma}

\begin{proof}
   Suppose that between time $t_1+1$ and $t_2$ there are requests made at points $x_1,x_2,\ldots,x_{k+k'} \in X$
   that satisfy $d(x_i, x_j) > D$ for all $i \neq j$.  Then clearly:
            \[
                     \cost^*_X(\sigma_{[1,t_2]}) \geq \cost^*_X(\sigma_{[1,t_1]}) + k' D\,.\qedhere
            \]
\end{proof}

\begin{theorem}
   For every time $t \geq 1$:
   \[
      \E_{\Omega_X} \left[\cost_{\hst_{\tau}}\!\left((F_{\bm{\cP}}^{\otimes k} \circ \opt^X)(\sigma_{[1,t]})\right)\right] \leq O(M \log k)\, \cost_X^*(\sigma_{[1,t]}) + O_{X,k}(1)\,.
   \]
\end{theorem}

\begin{proof}
   We can split the movement of $F^{\otimes k}_{\bm{\cP}} \circ \opt^X$ into three parts:
   First the stack $\bm{\cP}$ is possibly updated by the algorithm, either with a reset or an insertion (each induce movement), and then we mirror the move of $\opt^X$. 

Let us first consider the case of a level-$j$ reset at time $t$. The key observation is that the mapping $F_{\bm{\cP_t}}$ remains identical to $F_{\bm{\cP_t}}$ on the $(j-1)$-prefix, that is $(P_t^0(x),\ldots, P_t^{j-1}(x)) =  (P_{t-1}^0(x),\ldots, P_{t-1}^{j-1}(x))$, $\forall x \in X$. Thus the movement cost induced by a level-$j$ reset is at most $k \tau^{1-j}$.
   In particular the total cost of resets up to time $t$ is upper bounded by
   \[
      \sum_{j=1}^M K_{j,t} \cdot k\tau^{1-j} \leq \tau^2 M \cdot \cost_X^*(\sigma_{[1,t]})\,,
   \]
   where the last inequality is \pref{lem:resets}.
\newline

For the cost resulting from a level-$j$ insertion at time $t$, we use the following argument.
Assume that $j$ is the smallest index in $[M]$ with a level-$j$ insertion. Note that, by construction of the radii,
 one has $P^h_t(\sigma_t) \subseteq P_t^j(\sigma_t)$ for $h \geq j$. In particular,
the movement comes from the set of servers $I_t=\{i \in [k] : (\opt^X_{t-1})_i \in P^j_t(\sigma_t)\}$, for which the mapping $F_{\bm{\cP_t}}$ will change the $j$-suffix (i.e., $(P_t^j, \ldots, P_t^M)$) compared to $F_{\bm{\cP_{t-1}}}$. However, importantly the $(j-1)$-prefix remains identical (this is because there is no insertion at level $j-1$ and thus these servers remain part of the same non-singleton cluster at level $j-1$).

In other words, the total movement is $O(\tau^{-j} |I_t|)$. Now we argue that we can match this movement with either reset movement or movement coming from $\opt^X$ as follows. First we ignore the set of servers in $J_t \subseteq I_t$ such that their $(j-1)$-prefix remain the same forever, indeed one has $\sum_{t \geq 1} |J_t| = O_{X,k}(1)$. Now for a server $i \in I_t \setminus J_t$ consider the first time $s\geq t$ such that the $(j-1)$-prefix of $F_{\bm{\cP_{t-1}}}((\opt^X_{s})_i)$ is different from $F_{\bm{\cP_{t}}}((\opt^X_{t-1})_i)$. The corresponding movement at time $s$ comes either from a reset or from a movement of $\opt^X$, and moreover its cost is larger than $\tau^{-j}$. Thus we just showed that at the expense of an additive term
of $O_{X,k}(1)$ and a multiplicative factor $2$ for movement cost induced by resets and the movement of $\opt^X$,
one can ignore the movement cost induced by insertions.
\newline

  Finally, we deal with the cost coming from movement of $\opt^X$. We may assume that $\opt^X$ is conservative:  If $\sigma_t \in \opt^X_{t-1}$, then $\opt^X_t=\opt^X_{t-1}$
   and otherwise $\opt^X_{t-1} \setminus \opt^X_{t} = \{x_t\}$ for some $x_t \in X$, and $\opt^X$ pays $d(\sigma_t,x_t)$.
   From \pref{lem:stretch}, the expected cost of mirroring this move in $\hst_{\tau}$ is at most $O(M \log k)\, d(\sigma_t,x_t)$.
\end{proof}

Thus we are left only to analyze the stretch.
Note that since $d(\sigma_t, N_t^j) \leq \tau^{-j-1}$ by construction,
the next lemma yields \pref{lem:stretch} when combined with \pref{lem:prestack}.

\begin{lemma}
   For every $y \in X$ satisfying $d(y,N_t^j) \leq \tau^{-j-1}$ and every $x \in X$:
   \[
      \Pr\left[P^j_t(x) \neq P^j_t(y) \right] \leq (2+4e) d(x,y) \tau^{j+1} \log k\,.
   \]
\end{lemma}

\begin{proof}
   Note that $|N_t^j| \leq 2k$ by construction.  Let us arrange the centers
   in the order which they were added: $N_t^j = \{ x_1, x_2, \ldots, x_N \}$.
   Let $\hat{R}_1,\ldots,\hat{R}_N$ denote the corresponding random radii $\hat{R}_i \seteq \hat{R}_t^j(x_i)$
   and let $R_i\seteq \hat{R}_i + \tau^{-j-1}$.

   Denote the event
   \begin{align*}
      \cE_i &\seteq \left\{d(x_i, \{x,y\}) \leq R_i \wedge \max \{d(x_i,x),d(x_i,y)\} > R_i\right\},
   \end{align*}
   and let $i_* \seteq \min \{ i : d(x_i, \{x,y\}) \leq R_i \}$.
   Define $c \seteq (1-\frac{1}{\log k})\tau^{-j-1}$.  Then:
   \begin{align*}
      \Pr\left[P_t^j(x) \neq P_t^j(y)\right] &= \sum_{i=1}^{N} \Pr[i = i_*] \cdot \Pr[\cE_i \mid i=i_*] \\
                                             &\leq \sum_{i=1}^N \Pr[\cE_i \wedge \{\hat{R}_i \geq c\}] + \sum_{i=1}^N \Pr[i=i_*] \cdot \Pr\left[\cE_i \wedge \{\hat{R}_i < c\} \mid i=i_*\right]\,.
      \end{align*}
      For any $i=1,\ldots,N$, we have
   \begin{align*}
      \Pr[\cE_i \wedge \{\hat{R}_i \geq c\}] &\leq \sup_{R \geq c} \left\{\int_{R}^{R+d(x,y)} d\mu_{j+1}(r) \right\}\\
      &\leq \int_{c}^{c+d(x,y)} d\mu_{j+1}(r) \\
      &= \frac{k}{k-1} \left(1-\exp\left(-d(x,y) \tau^{j+1} \log k\right)\right) e^{-c \tau^{j+1} \log k} \\
      &\leq \frac{2 e}{k} d(x,y)\tau^{j+1} \log k\,,
   \end{align*}
   where the final line uses $1-e^{-u} \leq u$ and $k \geq 2$.  This yields $\sum_{i=1}^N \Pr[\cE_i \wedge \{\hat{R}_i \geq c\}] \leq 4e d(x,y) \tau^{j+1} \log k$ since
   $N \leq 2k$.

   Now analyze:
   \begin{align*}
      \Pr[\cE_i \wedge \{\hat{R}_i < c\} \mid i=i_*] &\leq \sup_{0 \leq R < c} \left\{ \frac{\int_{R}^{R+d(x,y)} d\mu_{j+1}(r)}{\int_{R}^{\tau^{-j-1}} d\mu_{j+1}(r)} \right\} \\
         &\leq \frac{\int_{c}^{c+d(x,y)} d\mu_{j+1}(r)}{\int_{c}^{\tau^{-j-1}} d\mu_{j+1}(r)}  \\
      &= \frac{1-\exp\left(-d(x,y) \tau^{j+1} \log k\right)}{1-\exp(c \tau^{j+1} \log k)/k}
      \leq 2 d(x,y) \tau^{j+1} \log k \,,
   \end{align*}
   where in the last line we have used again $1-e^{-u} \leq u$ and the value of $c$.
\end{proof}

\section{Mirror descent}
\label{sec:mirror}

\newcommand{\spe}{*}

We now prove \pref{thm:mirror-intro}.

\subsection{Preliminaries}
Consider an $\R^n$-set-valued map $F$ with domain $X \subseteq \R^{n}$. We will be interested in the following {\em viability problem} (i.e., a differential inclusion with a constraint set for the solution): Given a constraint set $K \subseteq X$ and initial point $x_{0}\in K$, find an absolutely continuous solution $x : [0,\infty) \rightarrow K$ such that 
\begin{align*}
\partial_t x(t) & \in F(x(t))\,,\\
x(0) & =x_{0}\,.
\end{align*}
The upshot is that, under appropriate continuity condition on $F$, this problem has a solution provided that $F$ always contain {\em admissible directions}, that is $F(x)\cap T_K(x) \neq\emptyset$ where $T_K(x)$ is the tangent cone to $K$ at $x$ (see definition below). We now recall the needed definitions with some basic results, and state the general existence theorem from \cite{AC84}.

\begin{definition}
The polar $N^{\circ}$ of a set $N \subseteq \R^n$ is
\[N^{\circ} \seteq \{z \in \R^n : \langle z, y \rangle \leq 0 \; \text{for all} \; y \in N \} \,.\]
The {\em normal cone to $K$ at $x$} is
$$N_K(x) = (K-x)^{\circ} ~,$$
and the {\em tangent cone to $K$ at $x$} is
\[T_K(x) \seteq N_K(x)^{\circ}\,.\]
\end{definition}

\begin{lemma}[{Moreau's decomposition \cite[Thm III.3.2.5]{HUL93}}] \label{lem:moreau}
Let $N$ be a cone. Then for any $x \in \R^n$ there is a unique pair $(u,v) \in N \times N^{\circ}$ such that $\langle u, v \rangle = 0$ and $x=u+v$. Furthermore $u$ (resp., $v$) is the projection of $x$ onto $N$ (resp., $N^{\circ}$).
\end{lemma}

\begin{definition}
   $F$ is {\em upper semicontinuous (u.s.c.)} if for any $x \in X$ and any open neighborhood $N \supset F(x)$ there exists a neighborhood $M$ of $x$ such that $F(M) \subseteq N$. $F$ is {\em upper hemicontinuous (u.h.c.)} if, for any $\theta \in \R^n$, the support function $x \mapsto \sup_{y \in F(x)} \langle \theta, y \rangle$ is upper semicontinuous.
\end{definition}

\begin{lemma}[{\cite[Prop. 1, pg. 60; Cor. 2, pg. 63]{AC84}}] \label{lem:uscuhc}
   For any $F$ as above,
u.s.c. implies u.h.c., and moreover if $F$ takes compact, convex values then the two notions are equivalent.
\end{lemma}

\begin{lemma}[{\cite[Thm. 1, pg. 41]{AC84}}] \label{lem:inter}
Let $F$ and $G$ be two set valued maps such that $F$ is u.s.c., $F$ takes compact values, and the graph of $G$ is closed. Then
the set-valued mapping $x \mapsto F(x) \cap G(x)$ is u.s.c.
\end{lemma}

\begin{theorem}[{\cite[Thm. 1, pg. 180]{AC84}}]
\label{thm:viability}
Assume that $F$ is u.h.c. and takes compact, convex values, and that $K$ is compact. Furthermore assume the tangential condition: For any $x\in X$, 
\[F(x)\cap T_K(x) \neq\emptyset\,.\]
Then the viability problem admits an absolutely continuous solution.
\end{theorem}

\subsection{Existence}

We prove the following theorem which can be viewed as a non-Euclidean extension of \cite[pg. 217]{AC84}.

\begin{theorem} \label{thm:mirror_exists}
Let $K \subseteq \R^n$ be a compact convex set, let $H : K \rightarrow \{A \in \R^{n \times n} : A \succ 0\}$ be continuous, and let $f:[0,\infty)\times K\rightarrow\R^n$ be continuous.
Then, for any $x_{0}\in K$,
there is an absolutely continuous solution $x : [0,\infty) \to K$
$K$ satisfying 
\begin{align}
\partial_t x(t) & \in H(x) \left(f(t,x(t))-N_{K}(x(t))\right)\,,\label{eq:mirror_descent}\\
x(0) & =x_{0}\,.\nonumber 
\end{align}
Furthermore, any solution to the viability
problem satisfies
\begin{equation}
\partial_{t}x(t)=\argmin \left\{\norm{v-H(x)f(t,x(t))}_{x,*}^{2}: v\in T_{K}(x(t))\right\}.\label{eq:least_action}
\end{equation}
In particular, one has
\begin{equation}
\|\partial_t x(t)\|_{x,*} \leq \|f(t,x)\|_x\,.\label{eq:norm_smaller}
\end{equation}

\end{theorem}

\begin{proof}
It suffices to prove the existence on any time interval $[T,T+1]$. We denote $\langle \cdot, \cdot \rangle_x$ for the inner product induced by $H(x)$ (i.e., $\langle \alpha, \beta \rangle_x := \langle \alpha, H(x) \beta\rangle$), $\|\cdot\|_x$ for the corresponding norm, and $\|\cdot\|_{x,*}$ for its dual norm. To apply \pref{thm:viability}, consider the following differential inclusion, with $\overline{K}=[T,T+1] \times K$,
\[
\widetilde{x}'(t)\in F(\widetilde{x}(t))\,,
\]
where $F:\overline{K}\rightarrow2^{\R^{n+1}}$ defined by
\begin{align*}
   F(t,x) & \seteq \left(\vphantom{\bigoplus} 1,H(x)\left(f(t,x)-N_{K}(x)) \cap \{v \in \R^n : \|v\|_{x,*} \leq \|f(t,x)\|_x\}\right)\right)\,.
\end{align*}
Thanks to the restriction to velocities satisfying $ \|v\|_{x,*} \leq \|f(t,x)\|_x$, it holds that $F(t,x)$ is compact (it is also clearly convex). Moreover, $\overline{K}$ is compact. Thus, besides the tangential condition, it remains to show that $F$ is u.h.c. Since $F$ is compact and convex valued, by \pref{lem:uscuhc} it suffices to show that $F$ is u.s.c. For this we apply \pref{lem:inter}. Note that $\{(1, H(x)(f(t,x)-y)) : x \in K, y \in N_K(x)\}$ is closed (using again continuity of $H$ and $f$), and thus it suffices to show that the mapping
$x \mapsto \{v \in \R^n : \|v\|_{x,*} \leq \|f(t,x)\|_x\}$ is u.s.c., which is clearly true since its support function in direction $\theta$ is $x \mapsto \|f(t,x)\|_x \|\theta\|_x$ which is continuous by continuity of $H$ and $f$.

It remains to check the tangential condition.
We note that
\begin{align*}
   N_{\overline{K}}(t,x) &=\begin{cases}
(-\infty,0]\times N_{K}(x) & t=0\,,\\
\{0\}\times N_{K}(x) & \mathrm{otherwise.}
\end{cases}  \\
T_{\overline{K}}(t,x) &=\begin{cases}
[0,\infty)\times T_K(x) & t=0\,,\\
\R\times T_{K}(x) & \mathrm{otherwise.}
\end{cases}
\end{align*}
Thus it suffices to show that there exists $u \in N_K(x)$ such that $\|H(x)(f(t,x)-u)\|_{x,*} \leq \|f(t,x)\|_x$ and $H(x) (f(t,x) - u) \in T_K(x)$.

We apply Moreau's decomposition (\pref{lem:moreau}) in $\R^n$ equipped with the inner product $\langle \cdot, \cdot \rangle_x$ to the cone $N_K(x)$ and write $f(t,x) = u + v$ where $u \in N_K(x)$ and $v$ is in the polar (w.r.t. $\langle \cdot, \cdot \rangle_x$) of $N_K(x)$.
Since $\langle u, v \rangle_x = 0$, we have $\|f(t,x)-u\|_x \leq \|f(t,x)\|_x$. This gives $\|H(x)(f(t,x)-u)\|_{x,*} \leq \|f(t,x)\|_x$. Furthermore since $v$ is in the polar,
we have for all $y \in N_K(x), \langle f(t,x) -u , y \rangle_x \leq 0$ which means that $H(x) (f(t,x) - u) \in T_K(x)$. This concludes the existence proof.

Since the objective function in (\ref{eq:least_action}) is strongly
convex, the optimality condition of (\ref{eq:least_action}) shows
that any $v\in T_{K}(x(t))$ satisfying $H(x)^{-1}(v-H(x)f(t,x(t)))\in-T_{K}(x(t))^{\circ}$
is the unique solution. Now, we note that $\partial_{t}x(t)\in T_{K}(x(t))$
because $x(t)\in K$ and
\[
H(x)^{-1}(\partial_{t}x(t)-H(x)f(t,x(t)))\in-N_{K}(x(t))=-T_{K}(x(t))^{\circ}\,,
\]
where we used that $T_{K}(x(t))^{\circ}=N_{K}(x(t))^{\circ\circ}=N_{K}(x(t))$
(since $N_{K}(x(t))$ is a closed and convex cone). This establishes (\ref{eq:least_action}).

Since we exhibited a solution satisfying \eqref{eq:norm_smaller}, the almost
everywhere uniqueness of $\partial_t x(t)$ shows that \eqref{eq:norm_smaller} holds for any solution.
\end{proof}

\subsection{Uniqueness}
We prove here that the solution to the viability problem is in fact unique under slightly more restrictive assumptions than those in \pref{thm:mirror_exists}.
In what follows,
we denote, as in the proof of \pref{thm:mirror_exists}, $\langle \cdot, \cdot \rangle_x$ for the inner product induced by $H(x)$, $\|\cdot\|_x$ for its corresponding norm and $\|\cdot\|_{x,*}$ for its dual norm. 

\begin{lemma}
\label{lem:xt_boundary}
Let $x(t)$ be an absolutely continuous path with values in a convex set $K \subseteq \R^n$. Then, $\partial_t x(t)\in N_{K}(x(t))^{\perp}$
almost everywhere.
\end{lemma}

\begin{proof}
For any $t$ such that $\partial_t x(t)$ exists, one has
\[
x(t+h)=x(t)+h \partial_t x(t)+o(h) ~.
\]
In particular for any $y \in N_{K}(x(t))$, since $x(t+h)\in K$,
one has
\[
\left\langle y,x(t+h)-x(t)\right\rangle \leq0.
\]
Taking $h\rightarrow0^{+}$, we obtain $\left\langle y, \partial_t x(t)\right\rangle \leq0$.
Taking $h\rightarrow0^{-}$, we obtain $\left\langle y, \partial_t x(t)\right\rangle \geq0$.
Therefore, $\partial_t x(t)\in N_{K}(x(t))^{\perp}$.
\end{proof}

\begin{lemma}
The solution $x(t)$ in \pref{thm:mirror_exists} is unique provided that $H$ is Lipschitz, and $f$ is locally Lipschitz.
\end{lemma}

\newcommand{\tx}{\tilde{x}}
\newcommand{\tu}{\tilde{u}}

\begin{proof}
Let $x(t)$ and $\tx(t)$ be two solutions to the viability problem. We show that for any $T \geq 0$ there are some constants $C, \epsilon>0$ such that for all $t \in [T, T+\epsilon]$,
\begin{equation} \label{eq:toproveforuniqueness}
\partial_t \|x(t) - \tx(t)\|_{x(t),*}^2 \leq C \|x(t) - \tx(t)\|_{x(t),*}^2\,,
\end{equation}
which concludes the proof by a simple application of Gronwall's inequality (notice that since $H$ is continuous and
$K$ is compact, there is some constant $\epsilon'$ such that $H^{-1}(x) \succeq \epsilon' \mathrm{I}_n,  \forall x \in K$).

Recall that we have $\partial_t x(t) = H(x(t)) (f(t,x(t)) - u(t))$ for some $u(t) \in N_K(x(t))$ (and similarly for $\tx$ there is some path $\tu$ in the normal cones). In particular we get:
\begin{align}
   \frac12 \partial_t \|x(t) - \tx(t)\|_{x(t),*}^2 &= \langle x(t) - \tx(t) ,  H(x(t)) (f(t,x(t)) - u(t)) - H(\tx(t)) (f(t,\tx(t)) - \tu(t)) \rangle_{x(t), *} \label{eq:tobound} \\
& \qquad + \frac{1}{2} (x(t) - \tx(t))^{\top} (\partial_t H(x(t))^{-1}) (x(t) - \tx(t))\,. \label{eq:tobound2}
\end{align}
Denote $\|\cdot\|_{\mathrm{op}}$ for the spectral norm, and observe that by continuity of $H$ and compactness of $K$, there exists $M>0$ such that for all $t$, $\|H(x(t))\|_{\mathrm{op}}, \|H(x(t))^{-1}\|_{\mathrm{op}} \leq M$. Thus we can bound the term \eqref{eq:tobound2} as follows:
\begin{align*}
(x(t) - \tx(t))^{\top} (\partial_t H(x(t))^{-1}) (x(t) - \tx(t)) & \leq \|x(t) - \tx(t)\|_{x(t),*}^2 \|H(x(t))^{1/2} (\partial_t H(x(t))^{-1}) H(x(t))^{1/2} \|_{\mathrm{op}} \\
& \leq M \|x(t) - \tx(t)\|_{x(t),*}^2 \|(\partial_t H(x(t))^{-1}) \|_{\mathrm{op}} ~.
\end{align*}

Now since $H$ is Lipschitz, there exists a constant $M'$ such that for all $t \in [T,T+\e]$,
\begin{align*}
\|(\partial_t H(x(t))^{-1}) \|_{\mathrm{op}} & = \|H(x(t))^{-1} (\partial_t H(x(t))) H(x(t))^{-1} \|_{\mathrm{op}} \\
& \leq M^2 M' \|\partial_t x(t)\|_{2} \leq M^3 M' \|\partial_t x(t)\|_{x(t),*} \leq M^3 M' \|f(t,x(t))\|_{x(t)}\,,
\end{align*}
where the last inequality follows from \eqref{eq:norm_smaller}. Since $f(t,x)$ is uniformly bounded on $[T, T+\epsilon]$, we finally get that for some $C>0$,
the term \eqref{eq:tobound2} is bounded from above by $C \|x(t) - \tx(t)\|_{x(t),*}^2$ for all $t \in [T+\e]$. 

Next we consider the term \eqref{eq:tobound} and decompose it into two terms:
\begin{quote}
\begin{enumerate}
   \item[(i)] $(x(t) - \tx(t))^{\top} \left(f(t,x(t)) - u(t) - (f(t,\tx(t)) - \tu(t))\right)$,  and 
   \item[(ii)] $\llangle x(t) - \tx(t) ,  \left(H(x(t)) - H(\tx(t))\right) \left(f(t,\tx(t)) - \tu(t)\right) \rrangle_{x(t), *}$.
\end{enumerate}
\end{quote}
We further decompose (i) into
\begin{quote}
\begin{enumerate}
   \item[(iii)] $(x(t) - \tx(t))^{\top} (f(t,x(t)) - f(t,\tx(t)))$, and
   \item[(iv)] $(x(t) - \tx(t))^{\top} (\tu(t) - u(t))$.
\end{enumerate}
\end{quote}
      To bound (iii), we now use that $f$ is locally Lipschitz,
      which shows that this term is bounded from above by $C \|x(t) - \tx(t)\|_{x(t),*}$ for any $t \in [T, T+\epsilon]$.
      
      The term (iv) is nonpositive since $u(t) \in N_K(x(t))$ and $\tu(t) \in N_K(\tx(t))$.
      Finally, for (ii) one can combine Lipschitzness of $H$ and \eqref{eq:norm_smaller} as above to obtain that it is bounded by $C \|x(t) - \tx(t)\|_{x(t),*}$ for all $t \in [T,T+\e]$, thus concluding the proof.
\end{proof}
 
\subsection*{Acknowledgements}

Part of this work was carried out while
M. B. Cohen, Y. T. Lee, and J. R. Lee were
at Microsoft Research in Redmond.
We thank Microsoft for their hospitality.
M. B. Cohen and A. M\k{a}dry are supported by
NSF grant CCF-1553428 and an Alfred P. Sloan Fellowship.
J. R. Lee is supported by NSF grants CCF-1616297 and CCF-1407779
and a Simons Investigator Award.

\bibliographystyle{alpha}
\bibliography{kserver}

\newcommand{\etalchar}[1]{$^{#1}$}
\begin{thebibliography}{BBMN15}

\bibitem[ABBS14]{ABBS10}
Jacob Abernethy, Peter Bartlett, Niv Buchbinder, and Isabelle Stanton.
\newblock A regularization approach to metrical task systems.
\newblock In {\em Algorithmic Learning Theory}, ALT 2010. Springer, 2014.

\bibitem[AC84]{AC84}
Jean~Pierre Aubin and A.~Cellina.
\newblock {\em Differential Inclusions: Set-Valued Maps and Viability Theory}.
\newblock Springer-Verlag New York, Inc., Secaucus, NJ, USA, 1984.

\bibitem[Bar96]{Bar96}
Yair Bartal.
\newblock Probabilistic approximations of metric spaces and its algorithmic
  applications.
\newblock In {\em 37th Annual Symposium on Foundations of Computer Science,
  {FOCS} '96, Burlington, Vermont, USA, 14-16 October, 1996}, pages 184--193,
  1996.

\bibitem[Bar98]{Bar98}
Yair Bartal.
\newblock On approximating arbitrary metrices by tree metrics.
\newblock In {\em STOC '98 (Dallas, TX)}, pages 161--168. ACM, New York, 1998.

\bibitem[BBK99]{BBK99}
Avrim Blum, Carl Burch, and Adam Kalai.
\newblock Finely-competitive paging.
\newblock In {\em 40th {A}nnual {S}ymposium on {F}oundations of {C}omputer
  {S}cience ({N}ew {Y}ork, 1999)}, pages 450--457. IEEE Computer Soc., Los
  Alamitos, CA, 1999.

\bibitem[BBM06]{BBM06}
Yair Bartal, B\'ela Bollob\'as, and Manor Mendel.
\newblock Ramsey-type theorems for metric spaces with applications to online
  problems.
\newblock {\em J. Comput. System Sci.}, 72(5):890--921, 2006.

\bibitem[BBMN15]{BBMN15}
Nikhil Bansal, Niv Buchbinder, Aleksander Madry, and Joseph Naor.
\newblock A polylogarithmic-competitive algorithm for the {$k$}-server problem.
\newblock {\em J. ACM}, 62(5):Art. 40, 49, 2015.

\bibitem[BBN12]{BBN12}
Nikhil Bansal, Niv Buchbinder, and Joseph Naor.
\newblock A primal-dual randomized algorithm for weighted paging.
\newblock {\em J. ACM}, 59(4):Art. 19, 24, 2012.

\bibitem[BCN14]{BCN14}
Niv Buchbinder, Shahar Chen, and Joseph~(Seffi) Naor.
\newblock Competitive analysis via regularization.
\newblock In {\em Proceedings of the Twenty-fifth Annual ACM-SIAM Symposium on
  Discrete Algorithms}, SODA '14, pages 436--444, Philadelphia, PA, USA, 2014.
  Society for Industrial and Applied Mathematics.

\bibitem[BE98]{BE98}
Allan Borodin and Ran {El-Yaniv}.
\newblock {\em Online computation and competitive analysis}.
\newblock Cambridge University Press, New York, 1998.

\bibitem[BLMN05]{BLMN05}
Yair Bartal, Nathan Linial, Manor Mendel, and Assaf Naor.
\newblock On metric {R}amsey-type phenomena.
\newblock {\em Ann. of Math. (2)}, 162(2):643--709, 2005.

\bibitem[Bub15]{Bub15}
S{\'e}bastien Bubeck.
\newblock Convex optimization: Algorithms and complexity.
\newblock {\em Foundations and Trends in Machine Learning}, 8(3-4):231--357,
  2015.

\bibitem[CMP08]{CMP08}
Aaron Cot\'e, Adam Meyerson, and Laura Poplawski.
\newblock Randomized {$k$}-server on hierarchical binary trees.
\newblock In {\em S{TOC}'08}, pages 227--233. ACM, New York, 2008.

\bibitem[FKL{\etalchar{+}}91]{FKL91}
A.~Fiat, R.~M. Karp, M.~Luby, L.~A. McGeoch, D.~D. Sleator, and Young~N. E.
\newblock Competitive paging algorithms.
\newblock {\em J. Algorithms}, 12(4):685--699, 1991.

\bibitem[FRT04]{FRT04}
Jittat Fakcharoenphol, Satish Rao, and Kunal Talwar.
\newblock A tight bound on approximating arbitrary metrics by tree metrics.
\newblock {\em J. Comput. Syst. Sci.}, 69(3):485--497, 2004.

\bibitem[Haz16]{Haz16}
Elad Hazan.
\newblock Introduction to online convex optimization.
\newblock {\em Foundations and Trends® in Optimization}, 2(3-4):157--325,
  2016.

\bibitem[HUL93]{HUL93}
Jean-Baptiste Hiriart-Urruty and Claude Lemar{\'e}chal.
\newblock {\em Convex analysis and minimization algorithms, Volume I:
  Fundamentals.}
\newblock Springer-Verlag, 1993.

\bibitem[Kou09]{Kou09}
Elias Koutsoupias.
\newblock The {$k$}-server problem.
\newblock {\em Computer Science Review}, 3(2):105--118, 2009.

\bibitem[KP95]{KP95}
Elias Koutsoupias and Christos~H. Papadimitriou.
\newblock On the {$k$}-server conjecture.
\newblock {\em J. Assoc. Comput. Mach.}, 42(5):971--983, 1995.

\bibitem[MMS90]{MMS90}
Mark~S. Manasse, Lyle~A. McGeoch, and Daniel~D. Sleator.
\newblock Competitive algorithms for server problems.
\newblock {\em J. Algorithms}, 11(2):208--230, 1990.

\bibitem[MS91]{MS91}
Lyle~A. McGeoch and Daniel~D. Sleator.
\newblock A strongly competitive randomized paging algorithm.
\newblock {\em Algorithmica}, 6(6):816--825, 1991.

\bibitem[NY83]{NY83}
A.~Nemirovski and D.~Yudin.
\newblock {\em Problem Complexity and Method Efficiency in Optimization}.
\newblock Wiley Interscience, 1983.

\end{thebibliography}

\end{document}